\documentclass[amsthm]{autart}

\usepackage{amsmath,amssymb,mathtools}
\usepackage{enumitem}

\theoremstyle{plain}
\newtheorem{theorem}{Theorem}
\newtheorem{lemma}{Lemma}
\newtheorem{corollary}{Corollary}
\newtheorem{proposition}{Proposition}

\theoremstyle{definition}
\newtheorem{definition}{Definition}
\newtheorem{assumption}{Assumption}
\newtheorem{remark}{Remark}

\newcommand{\R}{\mathbb{R}}

\newcommand{\Tcal}{\mathcal{T}}
\newcommand{\Linf}{L^\infty}
\newcommand{\Ltwo}{L^2}
\newcommand{\Lipk}{\mathrm{Lip}_K}
\newcommand{\LipKn}{\mathrm{Lip}_{K_N}}
\newcommand{\LipT}{\mathrm{Lip}_{\Tcal^{-1}}}
\newcommand{\Bball}{\overline{B}}
\newcommand{\eps}{\varepsilon}

\usepackage[hidelinks,bookmarks=false]{hyperref}

\allowdisplaybreaks
\begin{document}

\begin{frontmatter}

\title{
Approximate Feedback Linearization for a Nonlinear Hyperbolic PDE Class --- Part I: Volterra Truncation}

\author{Miroslav Krstic}
\ead{mkrstic@ucsd.edu}
\thanks{Department of Mechanical and Aerospace Engineering, University of California San Diego, La Jolla, CA 92093-0411.}
\thanks{The principal AI aid in developing the paper was Claude.}

\begin{abstract}
Backstepping for nonlinear PDEs yields exact feedback linearizing laws in the form of infinite Volterra series---elegant in theory, but with challenges for implementation. This paper shows that even very low-order truncations of such controllers, no longer exactly linearizing, retain the stabilizing power. The key insight is that higher-order terms become negligible near the origin, so stability is recovered for any fixed truncation order by restricting the initial condition size. We establish spatial sup-norm results: finite-time practical stability and asymptotic stability characterized by a class-$\mathcal{KL}$ estimate. The region-of-attraction estimate grows with the truncation order and shrinks with the growth rate of the nonlinearity. The analysis overcomes the lack of pointwise kernel bounds and resolves well-posedness of the nonlinear closed loop, showing that surprisingly simple approximations already capture the essence of nonlinear PDE feedback linearization.
\end{abstract}

\date{}
\end{frontmatter}

\section{Introduction}\label{sec:intro}

This is the first of two companion papers on approximate feedback linearization of a class of nonlinear hyperbolic PDEs. Part~I, the present paper, replaces the exact infinite Volterra linearizer of~\cite{krstic2026feedbacklinearizationhyperbolicpdes} with a finite truncation and establishes closed-loop $\Linf$ stability of the truncated feedback. Part~II~\cite{krstic2026feedbacklinearizationhyperbolicpdes-NO} takes the resulting finite controller and replaces it with a learned neural-operator surrogate, eliminating the plant-specific kernel solve and the $N$-fold nested integration that would otherwise be required at every control update.

\subsection*{Previous results on boundary control of PDEs with Volterra nonlinearities}
Vazquez and Krstic~\cite{vazquez2008volterra1,vazquez2008volterra2} introduced a Volterra-series backstepping framework for boundary control of one-dimensional parabolic PDEs with analytic-type spatial nonlinearities, mapping the closed-loop plant into a stable target heat equation through an infinite Volterra transformation whose kernels are governed by parabolic PDEs on simplices of growing dimension. The result is local exponential stabilization in $L^2$ under the infinite Volterra feedback.

The first hyperbolic counterpart of \cite{vazquez2008volterra1,vazquez2008volterra2} is \cite{krstic2026feedbacklinearizationhyperbolicpdes}, where the parabolic kernel PDEs are replaced by first-order transport PDEs on growing simplices, solvable along characteristics, and the inverse transformation is constructed by a contraction mapping in $L^2$; closed-loop $L^2$ exponential stability follows. A coefficient-level alternative is also developed there, via a  noncommutative-formal-series generalization of the Volterra representations of Lesiak and Krener~\cite{1101898}, replacing the simplex transport PDEs by scalar ODEs solvable by quadrature.

In finite-dimensional nonlinear control, the idea of transforming a nonlinear system into a linear one modulo higher-order terms goes back to the approximate feedback linearization of Krener~\cite{Krener1984ApproximateLinearization} and Kang~\cite{Kang1994ApproximateLinearization}. The present work differs in motivation: exact feedback linearization is available here through an infinite Volterra operator, and the higher-order residual appears only when that operator is truncated for implementation.

\subsection*{Results of this paper}
This paper is a sequel to and extension of \cite{krstic2026feedbacklinearizationhyperbolicpdes}, which constructs an exact feedback-linearizing controller for the nonlinear hyperbolic Volterra PDE \eqref{eq:plant}--\eqref{eq:Fdef} as an infinite Volterra series in the state, and proves local exponential stabilization. The controller of \cite{krstic2026feedbacklinearizationhyperbolicpdes} is itself a Volterra series, whose kernels are obtained by solving an infinite cascade of first-order transport PDEs on simplices of unbounded dimension. While \cite{krstic2026feedbacklinearizationhyperbolicpdes} establishes existence and convergence of this controller, it is not meant to be implemented as literally written: a computer cannot evaluate infinitely many nested integrals in real time, nor precompute kernels by solving infinitely many PDEs on simplices of unbounded dimension. Truncation at some finite order is unavoidable.

Truncation breaks the closed-loop equivalence that underlies the analysis of \cite{krstic2026feedbacklinearizationhyperbolicpdes}. Under the exact controller the backstepping transformation maps the plant into a homogeneous transport equation, autonomous in the target and trivially well-posed. Any finite truncation instead leaves a nonzero residual at the target's boundary --- the discarded Volterra tail along the trajectory, which depends on the state through the inverse transformation --- so the closed loop becomes a genuinely coupled fixed-point problem in time. Well-posedness is no longer free, and forward invariance, practical stability, and asymptotic decay must be re-established with the coupling, and with the residual kept small, on a forward-invariant set fixed by the truncation order.

The contribution of the present paper is to carry out this analysis quantitatively in the sup-norm, characterizing the closed loop as local asymptotic stabilization in the standard $\mathcal{KL}$ sense with quantitative practical-stability bounds. Theorem~\ref{thm:main} establishes four properties --- forward invariance of a sup-norm ball, an all-time practical-stability bound, finite-time attractivity to a residual ball, and a class-$\mathcal{KL}$ asymptotic estimate. Because the residual decays geometrically in the truncation order, the exact-feedback stabilization of \cite{krstic2026feedbacklinearizationhyperbolicpdes} is recovered as $N \to \infty$.

The move to the sup-norm forces a different chain of estimates. The kernel bound imported from \cite{krstic2026feedbacklinearizationhyperbolicpdes} is an integral bound on simplices of growing dimension, whereas the truncation residual is pointwise and the contraction-mapping argument runs on sup-norm constants; no pointwise bound on the kernels is available. The bridge is a Cauchy--Schwarz estimate on the simplices, which turns the imported integral bound into an $L^1$ kernel bound, and thence a sup-norm bound on both residual and controller, without ever bounding the kernels pointwise.

A secondary contribution (Section~\ref{sec:L2}) is a proof of stability also in the $\Ltwo$ norm. Truncating the transformation along with the feedback closes the target's boundary exactly and turns the discarded tail into a static order-$(N{+}1)$ term on the same target. The resulting $\Ltwo$ theorem shows that, while stability holds in both norms, the sup-norm result is not only practically more relevant but also mathematically sharper.

Locality is essential. The contraction radius, its domain, and the admissible initial-condition norm are all bounded, and counterexamples to global stabilizability of nonlinear PDEs by boundary control~\cite{vazquez2008volterra1,vazquez2008volterra2} make this unavoidable. Yet the simulation in \cite{krstic2026feedbacklinearizationhyperbolicpdes} shows that even low-order truncations recover a substantial basin relative to the open-loop blow-up region, the third order already stabilizing initial conditions on which the second merely defers blow-up; the present paper formalizes the mechanism, each truncation order shrinking the residual geometrically and enlarging the region of attraction toward the full domain predicted by \cite{krstic2026feedbacklinearizationhyperbolicpdes}.

Finiteness of the feedback, however, does not make the computation light. Each update of $K_N$ needs the kernels $k_2, \ldots, k_N$ on simplices of dimension up to $N$ and an $N$-fold nested integration against them; and because the kernel PDEs are plant-specific, any change in the coefficients $f_2, \ldots, f_N$ forces a re-solve. Collapsing this online cost to a single learned map, while preserving the four stability properties proved here, is the task of the companion Part~II~\cite{krstic2026feedbacklinearizationhyperbolicpdes-NO}, which treats $K_N$ as one operator in plant coefficients and state and replaces it by a neural-operator surrogate.

\paragraph*{\normalfont\em Organization of the paper.}
Sections~\ref{sec:problem}--\ref{sec:mainresult} set up the plant and state the main theorem; Sections~\ref{sec:tail}--\ref{sec:transport} establish the supporting estimates; Section~\ref{sec:mainproof} proves the theorem; Section~\ref{sec:L2} gives the $\Ltwo$ alternative; and Section~\ref{sec:conclusions} concludes. The appendices collect the kernel construction and the operator-bound and $\Ltwo$ derivations.

\paragraph*{\normalfont\em Terminology.}
The \emph{slice} of $k_n$ at $x$ is the function $k_n(x,\cdot)$ on the $n$-simplex $T_n(x)$. The \emph{trace} of the controller is its boundary value $K[u](1)$, the \emph{trace kernel} at order $n$ being $k_n(1,\cdot)$. The \emph{truncation residual} (or \emph{tail}) is $K[u](1) - K_N[u](1)$. The \emph{transport crossing time} is $t = 1$, after which $w_0$ has left $[0,1]$ in exact feedback linearization; a \emph{sliding window} of width one is $[t-1, t]$.

\section{Problem formulation and standing results}\label{sec:problem}

\subsection{Plant and feedback}\label{sec:plant}

The plant is the nonlinear hyperbolic Volterra PDE
\begin{align}
u_t(x,t) &= u_x(x,t) + F[u](x,t), \quad x \in [0,1), \ t \ge 0, \label{eq:plant} \\
u(1,t) &= U(t), \label{eq:plant_bdy}\\
u(x,0) &= u_0(x), \label{eq:plant_ic}
\end{align}
where the nonlinearity $F[u]$ is the spatial Volterra series
\begin{eqnarray}
F[u](x,t) &=& \sum_{n=2}^\infty \int_{T_n(x)} f_n(x,\xi_1,\ldots,\xi_n) \nonumber \\
&& \quad \times \prod_{i=1}^n u(\xi_i,t)\, d\xi_n \cdots d\xi_1,
\label{eq:Fdef}
\end{eqnarray}
with simplex domains $T_n(x) := \{(\xi_1,\ldots,\xi_n) \in \R^n : 0 \le \xi_n \le \cdots \le \xi_1 \le x\}$.

The feedback-linearizing backstepping transformation, derived in \cite{krstic2026feedbacklinearizationhyperbolicpdes}, is $\Tcal[u] = u - K[u]$ with
\begin{eqnarray}
K[u](x,t) &=& \sum_{n=2}^\infty \int_{T_n(x)} k_n(x,\xi_1,\ldots,\xi_n) \nonumber \\
&& \quad \times \prod_{i=1}^n u(\xi_i,t)\, d\xi_n \cdots d\xi_1,
\label{eq:Kdef}
\end{eqnarray}
where the kernels $\{k_n\}_{n \ge 2}$ are defined in Appendix~\ref{app:kndef}, and the exact feedback law
\begin{eqnarray}
U(t) &=& K[u(\cdot,t)](1) \nonumber \\
&=& \sum_{n=2}^\infty \int_{T_n(1)} k_n(1,\xi)\, u(\xi_1,t)\cdots u(\xi_n,t)\, d\xi.
\label{eq:Uexact}
\end{eqnarray}
Under (\ref{eq:Uexact}), the target system is the pure transport
\begin{equation}
w_t = w_x, \qquad w(1,t) = 0, \qquad w(\cdot,0) = u_0 - K[u_0].
\label{eq:target_exact}
\end{equation}

\subsection{Standing assumption and imported result}\label{sec:assumption}

\begin{assumption}[Plant kernel growth]\label{ass:fnbound}
There exist constants $D_f, \rho_f > 0$ such that
\begin{equation}
\|f_n\|_{\Linf(T_n(1))} \le \frac{n!\, D_f}{\rho_f^{n-1}}, \qquad n \ge 2.
\label{eq:fnbound}
\end{equation}
\end{assumption}

The kernels $\{k_n\}_{n \ge 2}$ are defined in Appendix~\ref{app:kndef}. The following $\Ltwo$ slice bound, established in \cite{krstic2026feedbacklinearizationhyperbolicpdes}, is the only fact about $k_n$ used in this paper.

\begin{proposition}[$\Ltwo$ kernel slice bound, imported]\label{prop:L2bound}
Under Assumption~\ref{ass:fnbound}, the controller kernels $\{k_n\}_{n\ge 2}$ defined in Appendix~\ref{app:kndef} satisfy, for every $n \ge 2$ and every $x \in [0,1]$,
\begin{equation}
\|k_n(x,\cdot)\|_{\Ltwo(T_n(x))}^2 \;\le\; n!\, D_K^2\, C_K^{2(n-1)}\, x^n\, e^{2\Upsilon_K x},
\label{eq:knL2}
\end{equation}
for positive constants $D_K, C_K, \Upsilon_K$ depending only on $D_f$ and $\rho_f$. In particular, taking $x = 1$,
\begin{equation}
\|k_n(1,\cdot)\|_{\Ltwo(T_n(1))}^2 \;\le\; n!\, D_K^2\, C_K^{2(n-1)}\, e^{2\Upsilon_K}, \qquad n \ge 2.
\label{eq:knL2_endpoint}
\end{equation}
\end{proposition}

\subsection{Truncated feedback and constants}\label{sec:truncfeedback}

For an integer truncation order $N \ge 2$, define the order-$N$ truncated Volterra operator
\begin{eqnarray}
K_N[u](x,t) &:=& \sum_{n=2}^N \int_{T_n(x)} k_n(x,\xi_1,\ldots,\xi_n) \nonumber \\
&& \quad \times \prod_{i=1}^n u(\xi_i,t)\, d\xi_n \cdots d\xi_1,
\label{eq:KNdef}
\end{eqnarray}
which is the partial sum of (\ref{eq:Kdef}). The corresponding truncated feedback law is $U(t) = K_N[u](1,t)$.

For $r \in [0, 1/C_K)$, define
\begin{align}
\Lipk(r) &:= D_K\, e^{\Upsilon_K}\, \frac{C_K r\,(2 - C_K r)}{(1 - C_K r)^2}, \label{eq:LipKdef}\\
\eps_N(r) &:= \frac{D_K\, e^{\Upsilon_K}}{C_K}\, \frac{(C_K r)^{N+1}}{1 - C_K r}.
\label{eq:epsNdef}
\end{align}
The quantity $\Lipk(r)$ is the sup-norm Lipschitz constant of $K$ on the closed ball of radius $r$, and $\eps_N(r)$ is the magnitude of the truncation tail $K - K_N$ on that ball; the latter decays geometrically in $N$ at rate $(C_K r)^N$.

Then $\Lipk$ is continuous and strictly increasing on $[0, 1/C_K)$ with $\Lipk(0) = 0$ and $\Lipk(r) \uparrow \infty$ as $r \uparrow 1/C_K$, so the equation $\Lipk(r^*) = 1$ has a unique solution $r^* \in (0, 1/C_K)$. In particular, $r^* < 1/C_K$, so any $r_0 \in (0, r^*)$ satisfies $C_K r_0 < 1$ and $\Lipk(r_0) < 1$. To ensure that the plant nonlinearity (\ref{eq:Fdef}) is well-defined throughout the analysis, we further require $r_0 < \rho_f$, where $\rho_f$ is the constant of Assumption~\ref{ass:fnbound}; we therefore work with $r_0 \in (0, \bar r)$ for
\begin{equation}
\bar r \;:=\; \min(r^*,\, \rho_f).
\label{eq:rbardef}
\end{equation}

For any $r_0 \in (0, \bar r)$, define
\begin{eqnarray}
\rho_w(r_0) &:=& r_0\,\bigl(1 - \Lipk(r_0)\bigr), \label{eq:rhowdef}\\
\rho_u^0(r_0) &:=& r_0\,\frac{1 - \Lipk(r_0)}{1 + \Lipk(r_0)}, \label{eq:rhou0def}\\
\LipT(r_0) &:=& \frac{1}{1 - \Lipk(r_0)} \nonumber \\
&=& \frac{(1-C_K r_0)^2}{(1-C_K r_0)^2 - D_K e^{\Upsilon_K} C_K r_0 (2 - C_K r_0)}.
\label{eq:LipTinvdef}
\end{eqnarray}

\section{Main result}\label{sec:mainresult}

\begin{theorem}[Closed-loop $\Linf$ stability under truncated feedback]\label{thm:main}
Let Assumption~\ref{ass:fnbound} hold, and let $D_K, C_K, \Upsilon_K$ be the constants of Proposition~\ref{prop:L2bound}. Fix $r_0 \in (0, \bar r)$ and a truncation order $N \ge 2$ such that
\begin{equation}
\eps_N(r_0) \;<\; \rho_w(r_0).
\label{eq:cond_resid}
\end{equation}
Let $u_0 \in \Linf(0,1)$ satisfy
\begin{equation}
\|u_0\|_\infty \;<\; r_0\,\frac{1 - \Lipk(r_0)}{1 + \Lipk(r_0)}.
\label{eq:cond_init}
\end{equation}
Then the closed-loop system (\ref{eq:plant})--(\ref{eq:Fdef}) under the truncated feedback $U(t) = K_N[u](1,t)$ from (\ref{eq:KNdef}) admits a unique canonical mild solution $u \in \Linf_{\mathrm{loc}}([0,\infty); \Linf(0,1))$ with $u(\cdot, 0) = u_0$, satisfying:
\begin{enumerate}[label=(\roman*)]
\item \emph{(Forward invariance.)} $\|u(\cdot, t)\|_\infty \le r_0$ for every $t \ge 0$.
\item \emph{(Practical stability.)} For all $t \ge 0$,
\begin{eqnarray}
\|u(\cdot,t)\|_\infty &\le& \max\!\left(\frac{1 + \Lipk(r_0)}{1 - \Lipk(r_0)}\,\|u_0\|_\infty,\right. \nonumber \\
&& \quad\quad\;\, \left.\LipT(r_0)\,\eps_N(r_0)\right).
\label{eq:thm_ii}
\end{eqnarray}
\item \emph{(Practical finite-time attractivity.)} For all $t \ge 1$,
\begin{equation}
\|u(\cdot,t)\|_\infty \;\le\; \LipT(r_0)\, \eps_N(r_0).
\label{eq:thm_iii}
\end{equation}
\item \emph{(Asymptotic stability.)} There exists a function $\beta$ of class $\mathcal{KL}$ on $[0, \rho_u^0(r_0)) \times [0, \infty)$ such that
\begin{equation}
\|u(\cdot,t)\|_\infty \;\le\; \beta\bigl(\|u_0\|_\infty,\, t\bigr), \qquad t \ge 0.
\label{eq:thm_iv}
\end{equation}
\end{enumerate}
The residual in (ii) and (iii) decays geometrically in the truncation order: $\eps_N(r_0) = O\bigl((C_K r_0)^N\bigr)$ as $N \to \infty$, where $C_K < 1/r_0$. As $N \to \infty$, the practical-stability bound (\ref{eq:thm_ii}) collapses to the linear envelope $\frac{1+\Lipk(r_0)}{1-\Lipk(r_0)}\|u_0\|_\infty$, recovering the stabilization of the original system under exact (untruncated) feedback.
\end{theorem}

Theorem~\ref{thm:main} is a local asymptotic stabilization result in the $\Linf$ norm, carrying quantitative practical-stability bounds alongside the $\mathcal{KL}$ statement. It holds at every truncation order $N \ge 2$ meeting the residual-in-ball inequality (\ref{eq:cond_resid}) --- a mild condition, satisfied for moderate radii $r_0$ at the second and third orders used in practice --- with no auxiliary high-order requirement. The reason is that the truncation tail is higher-order in the state: it becomes innocuous on a small enough ball whatever the order, so a low $N$ is paid for by a smaller basin, never by loss of the conclusion.

That same higher-order vanishing is what reconciles the two decay statements, which otherwise look at odds. The tail $K[u] - K_N[u]$ collects orders $n \ge N+1$ and vanishes at the rate $\eps_N(\|u\|_\infty) = O(\|u\|_\infty^{N+1})$, so the perturbation driving the boundary residual $b(t)$ dies with the state; asymptotic stability of $u \equiv 0$ is not obstructed. What the uniform bounds (ii)--(iii) give up is not decay but time-uniformity: they replace $|b(t)|$ by its worst case $\eps_N(r_0)$ over the whole ball, buying an envelope that holds at every $t$ with explicit constants and depends only on the size of the state, not its history. Statement (iv) recovers the decay the worst-case bound discards, iterating the higher-order smallness along trajectories into a $\mathcal{KL}$ envelope --- at the cost of explicit constants in $\beta$. The $\beta$ so produced fits no single exponential rate: as the state shrinks the iteration accelerates, each unit-time step contracting the bound by a factor that itself tends to zero, so the decay is superexponential near the origin. When a computable envelope at finite $t$ is wanted, (ii) is the statement to quote; for the qualitative characterization of the equilibrium, (iv).

Based on \eqref{eq:cond_init}, \eqref{eq:rhou0def}, \eqref{eq:LipKdef}, \eqref{eq:rbardef}, \eqref{eq:fnbound}, and \eqref{eq:knL2}, the design guarantees asymptotic stability for initial conditions at least of the size $\|u_0\|_\infty < \sup_{0 < r < \bar r} \Lambda(r)$, where $\bar r := \min\{\rho_f,\, z_N/C_K\}$, $\Lambda(r) = r\,(1 - L(r))/(1 + L(r))$, $L(r) = A\,C_K r\,(2 - C_K r)/(1 - C_K r)^2$, and $A := D_K\, e^{\Upsilon_K}$. The gain $\Lambda$ vanishes at $r = 0$ and at $r = \tfrac{1}{C_K}(1 - \sqrt{A/(1+A)})$, with an interior maximum between; and $z_N(A)$ decreases in $A$ ($z_N(0) = 1$, $z_N(\infty) = 0$) and increases in $N$ ($z_\infty(A) < 1$). The basin thus grows with the truncation order $N$ and the analyticity radius $\rho_f$, and shrinks --- but never to zero, even at $N = 2$ --- as $N$ falls or as $A$ and $C_K$ rise.

The proof of Theorem~\ref{thm:main} is given in Section~\ref{sec:mainproof}, after establishing the supporting lemmas in Sections~\ref{sec:tail}--\ref{sec:transport}.

\section{Truncation tail bound}\label{sec:tail}

The truncated controller fails to feedback-linearize the plant exactly: it leaves a residual at the boundary equal to the discarded Volterra tail, and this residual is what every stability statement in the paper must accommodate. We show here that it decays geometrically in the truncation order, on any sup-norm ball strictly inside the radius $1/C_K$, by bounding a single high-order Volterra term and summing the geometric tail.

\begin{lemma}[Pointwise tail bound on Volterra terms]\label{lem:pointwise_tail}
Let Assumption~\ref{ass:fnbound} and Proposition~\ref{prop:L2bound} hold. For every $u \in \Linf(0,1)$ with $\|u\|_\infty < 1/C_K$ and every $n \ge 2$,
\begin{eqnarray}
&& \biggl|\int_{T_n(1)} k_n(1,\xi)\, u(\xi_1)\cdots u(\xi_n)\, d\xi\biggr| \nonumber \\
&& \quad \le\; \frac{D_K\, e^{\Upsilon_K}}{C_K}\,(C_K\,\|u\|_\infty)^n.
\label{eq:pointwise_tail}
\end{eqnarray}
\end{lemma}

\begin{proof}
By the Cauchy--Schwarz inequality on $T_n(1)$,
\begin{eqnarray}
&& \biggl|\int_{T_n(1)} k_n(1,\xi)\, u^{\otimes n}(\xi)\, d\xi\biggr| \nonumber \\
&& \quad \le\; \|k_n\|_{\Ltwo(T_n(1))}\, \|u^{\otimes n}\|_{\Ltwo(T_n(1))},
\label{eq:CS_step}
\end{eqnarray}
where 
\begin{eqnarray}
u^{\otimes n}(\xi) &:=& \prod_{i=1}^n u(\xi_i) \;=\; u(\xi_1)\cdots u(\xi_n), \nonumber \\
&& \quad \xi = (\xi_1,\ldots,\xi_n) \in T_n(1).
\label{eq:utensor}
\end{eqnarray}
Symmetrization gives
\begin{eqnarray}
\int_{T_n(1)} |u(\xi_1)\cdots u(\xi_n)|^2\, d\xi &=& \frac{1}{n!}\int_{[0,1]^n} \prod_{i=1}^n |u(\xi_i)|^2\, d\xi \nonumber \\
&=& \frac{\|u\|_{\Ltwo(0,1)}^{2n}}{n!},
\label{eq:symmetrization}
\end{eqnarray}
so $\|u^{\otimes n}\|_{\Ltwo(T_n(1))} = \|u\|_{\Ltwo(0,1)}^n / \sqrt{n!}$. Combining with the $\Ltwo$ kernel bound (\ref{eq:knL2_endpoint}),
\begin{eqnarray}
\biggl|\int_{T_n(1)} k_n\, u^{\otimes n}\, d\xi\biggr| &\le& \sqrt{n!}\, D_K\, C_K^{n-1}\, e^{\Upsilon_K}\cdot\frac{\|u\|_{\Ltwo}^n}{\sqrt{n!}} \nonumber \\
&=& \frac{D_K\, e^{\Upsilon_K}}{C_K}\,(C_K\,\|u\|_{\Ltwo})^n.
\label{eq:combined_bound}
\end{eqnarray}
Since $[0,1]$ has unit Lebesgue measure, $\|u\|_{\Ltwo(0,1)} \le \|u\|_{\Linf(0,1)}$, and (\ref{eq:pointwise_tail}) follows.
\end{proof}

The geometric tail starting at order $N+1$ now sums to the residual bound that the rest of the paper carries.

\begin{lemma}[Truncation tail]\label{lem:trunc_tail}
Under the hypotheses of Lemma~\ref{lem:pointwise_tail}, for every $r \in [0, 1/C_K)$, every $u \in \Linf(0,1)$ with $\|u\|_\infty \le r$, and every $N \ge 2$,
\begin{equation}
\bigl|K[u](1) - K_N[u](1)\bigr| \;\le\; \eps_N(r),
\label{eq:trunc_tail}
\end{equation}
where $K_N$ is defined by (\ref{eq:KNdef}) and $\eps_N(r)$ is given by (\ref{eq:epsNdef}).
\end{lemma}

\begin{proof}
By Lemma~\ref{lem:pointwise_tail} and $C_K\,\|u\|_\infty \le C_K r < 1$,
\begin{eqnarray}
\bigl|K[u](1) - K_N[u](1)\bigr| &=& \biggl|\sum_{n=N+1}^\infty \int_{T_n(1)} k_n\, u^{\otimes n}\, d\xi\biggr| \nonumber\\
&\le & \frac{D_K\, e^{\Upsilon_K}}{C_K} \sum_{n=N+1}^\infty (C_K r)^n \nonumber\\
&=& \frac{D_K\, e^{\Upsilon_K}}{C_K}\,\frac{(C_K r)^{N+1}}{1 - C_K r}.
\label{eq:tail_proof}
\end{eqnarray}
\end{proof}

\section{Sup-norm boundedness and Lipschitzness of $K$}\label{sec:KLip}

The kernel bound imported from \cite{krstic2026feedbacklinearizationhyperbolicpdes} is an $L^2$ bound on simplices whose dimension grows without bound, and the factorial growth on the right-hand side reflects the volume of those simplices. Direct sup-norm bounds on the kernels are not available, and pointwise bounds would not help: the truncation residual is a pointwise quantity, but the operator $K$ acts on $\Linf$ and what is needed is control of integrals of kernels against bounded states. The bridge is a single application of Cauchy--Schwarz on each simplex. The square root of the simplex volume, $1/\sqrt{n!}$, cancels exactly the $\sqrt{n!}$ on the right-hand side of the imported bound, and what remains is an $L^1$ kernel bound that is uniform in $n$ --- a geometric series in $C_K r$, summable on a non-trivial ball. Every sup-norm bound and every sup-norm Lipschitz constant in this paper is downstream of this single cancellation; the rest of the section just integrates against an $\Linf$ state to convert kernel bounds into operator bounds.

\begin{lemma}[$L^1$ kernel slice bound]\label{lem:knL1}
Under Assumption~\ref{ass:fnbound} and Proposition~\ref{prop:L2bound}, for every $n \ge 2$ and every $x \in [0,1]$,
\begin{eqnarray}
\|k_n(x,\cdot)\|_{L^1(T_n(x))} &\le& D_K\, C_K^{n-1}\, x^n\, e^{\Upsilon_K x} \nonumber \\
&\le& D_K\, C_K^{n-1}\, e^{\Upsilon_K}.
\label{eq:knL1}
\end{eqnarray}
\end{lemma}

\begin{proof}
By the Cauchy--Schwarz inequality on $T_n(x)$, and using $\mathrm{vol}(T_n(x))^{1/2} = x^{n/2}/\sqrt{n!}$,
\begin{eqnarray}
\|k_n(x,\cdot)\|_{L^1(T_n(x))} &\le& \|k_n(x,\cdot)\|_{\Ltwo(T_n(x))}\, \mathrm{vol}(T_n(x))^{1/2} \nonumber \\
&=& \|k_n(x,\cdot)\|_{\Ltwo(T_n(x))}\,\frac{x^{n/2}}{\sqrt{n!}}.
\label{eq:L1_CS}
\end{eqnarray}
Applying the slice bound (\ref{eq:knL2}),
\begin{eqnarray}
\|k_n(x,\cdot)\|_{L^1(T_n(x))} &\le& \sqrt{n!}\, D_K\, C_K^{n-1}\, x^{n/2}\, e^{\Upsilon_K x}\,\frac{x^{n/2}}{\sqrt{n!}} \nonumber \\
&=& D_K\, C_K^{n-1}\, x^n\, e^{\Upsilon_K x}.
\label{eq:L1_applied}
\end{eqnarray}
The second inequality in (\ref{eq:knL1}) follows from $x \le 1$.
\end{proof}

The $L^1$ kernel bound is what the rest of the section consumes: integrated against an $\Linf$ state, it pulls $\|u\|_\infty$ outside and leaves a geometric series.

\begin{lemma}[Sup-norm boundedness of $K$]\label{lem:Kbdd}
Under Assumption~\ref{ass:fnbound} and Proposition~\ref{prop:L2bound}, for every $u \in \Linf(0,1)$ with $\|u\|_\infty < 1/C_K$,
\begin{equation}
\|K[u]\|_\infty \;\le\; \frac{D_K\, e^{\Upsilon_K}}{C_K}\cdot \frac{(C_K\,\|u\|_\infty)^2}{1 - C_K\,\|u\|_\infty}.
\label{eq:Kbdd}
\end{equation}
\end{lemma}

\begin{proof}
Fix $x \in [0,1]$. Since $|u(\xi_1)\cdots u(\xi_n)| \le \|u\|_\infty^n$,
\begin{equation}
|K[u](x)| \;\le\; \sum_{n=2}^\infty \|k_n(x,\cdot)\|_{L^1(T_n(x))}\,\|u\|_\infty^n.
\label{eq:K_sum}
\end{equation}
Applying Lemma~\ref{lem:knL1},
\begin{eqnarray}
|K[u](x)| &\le& D_K\, e^{\Upsilon_K}\sum_{n=2}^\infty C_K^{n-1}\,\|u\|_\infty^n \nonumber \\
&=& \frac{D_K\, e^{\Upsilon_K}}{C_K}\sum_{n=2}^\infty (C_K\,\|u\|_\infty)^n \nonumber \\
&=& \frac{D_K\, e^{\Upsilon_K}}{C_K}\cdot\frac{(C_K\,\|u\|_\infty)^2}{1 - C_K\,\|u\|_\infty},
\label{eq:K_geom}
\end{eqnarray}
where the geometric series converges because $C_K\,\|u\|_\infty < 1$. Taking supremum over $x \in [0,1]$ gives (\ref{eq:Kbdd}).
\end{proof}

The Lipschitz estimate is the same calculation for a difference of states: the multilinear telescoping identity for $u_1^{\otimes n} - u_2^{\otimes n}$ produces an extra factor $n$, and the sum becomes the derivative of the previous one.

\begin{lemma}[Sup-norm Lipschitzness of $K$]\label{lem:KLip}
Under Assumption~\ref{ass:fnbound} and Proposition~\ref{prop:L2bound}, for every $r \in [0, 1/C_K)$ and every $u_1, u_2 \in \Linf(0,1)$ with $\|u_1\|_\infty, \|u_2\|_\infty \le r$,
\begin{equation}
\|K[u_1] - K[u_2]\|_\infty \;\le\; \Lipk(r)\,\|u_1 - u_2\|_\infty,
\label{eq:KLip}
\end{equation}
with $\Lipk(r)$ given by (\ref{eq:LipKdef}). In particular, $\Lipk(0) = 0$.
\end{lemma}

\begin{proof}
The multilinear telescoping identity
\begin{eqnarray}
&& \prod_{i=1}^n u_1(\xi_i) - \prod_{i=1}^n u_2(\xi_i) \nonumber \\
&=& \sum_{j=1}^n \biggl[\prod_{i<j} u_1(\xi_i)\biggr]\,\bigl[u_1(\xi_j) - u_2(\xi_j)\bigr] \nonumber \\
&& \quad \times \biggl[\prod_{i>j} u_2(\xi_i)\biggr]
\label{eq:telescope}
\end{eqnarray}
gives, on $T_n(x)$ with $\|u_1\|_\infty, \|u_2\|_\infty \le r$,
\begin{equation}
\bigl|u_1^{\otimes n}(\xi) - u_2^{\otimes n}(\xi)\bigr| \;\le\; n\, r^{n-1}\, \|u_1 - u_2\|_\infty.
\label{eq:telescope_bound}
\end{equation}
Therefore
\begin{eqnarray}
|K[u_1](x) - K[u_2](x)| &\le& \sum_{n=2}^\infty \|k_n(x,\cdot)\|_{L^1(T_n(x))} \nonumber \\
&& \quad \times n\, r^{n-1}\, \|u_1 - u_2\|_\infty.
\label{eq:K1K2_sum}
\end{eqnarray}
By Lemma~\ref{lem:knL1},
\begin{eqnarray}
|K[u_1](x) - K[u_2](x)| &\le& D_K\, e^{\Upsilon_K}\, \|u_1 - u_2\|_\infty \nonumber \\
&& \quad \times \sum_{n=2}^\infty n\, (C_K r)^{n-1}.
\label{eq:K1K2_applied}
\end{eqnarray}
For $y := C_K r \in [0,1)$, $\sum_{n=1}^\infty n\, y^{n-1} = (1-y)^{-2}$, so $\sum_{n=2}^\infty n\, y^{n-1} = (1-y)^{-2} - 1 = y(2-y)/(1-y)^2$. Hence
\begin{eqnarray}
\|K[u_1] - K[u_2]\|_\infty &\le& D_K\, e^{\Upsilon_K}\, \frac{C_K r\,(2 - C_K r)}{(1 - C_K r)^2} \nonumber \\
&& \quad \times \|u_1 - u_2\|_\infty \nonumber \\
&=& \Lipk(r)\,\|u_1 - u_2\|_\infty.
\label{eq:K1K2_final}
\end{eqnarray}
The identity $\Lipk(0) = 0$ follows from (\ref{eq:LipKdef}) and reflects the absence of a linear term in (\ref{eq:Kdef}), the series beginning at $n=2$.
\end{proof}

\section{Local inverse $\Tcal^{-1}$ and norm conversion}\label{sec:Tinv}

The backstepping transformation $\Tcal = I - K$ admits a local sup-norm inverse on a ball whose radius is governed by $\Lipk$. The technical statements (local invertibility, Lipschitz constant of $\Tcal^{-1}$) are derived in Appendix~\ref{app:Tinv} by a contraction-mapping argument that parallels the $\Ltwo$-norm construction in \cite{krstic2026feedbacklinearizationhyperbolicpdes}, with the same structure adapted to the sup-norm. The two end-products invoked in the closed-loop analysis are: a sup-norm Lipschitz bound for $\Tcal$ itself, and the conversion of a target-state bound to a plant-state bound via $\Tcal^{-1}$.

\begin{lemma}[Sup-norm Lipschitzness of $\Tcal$]\label{lem:TLip}
For every $r \in [0, 1/C_K)$ and every $u_1, u_2 \in \Linf(0,1)$ with $\|u_1\|_\infty, \|u_2\|_\infty \le r$,
\begin{equation}
\|\Tcal[u_1] - \Tcal[u_2]\|_\infty \;\le\; \bigl(1 + \Lipk(r)\bigr)\,\|u_1 - u_2\|_\infty.
\label{eq:TLip}
\end{equation}
\end{lemma}

\begin{proof}
Since $\Tcal[u_1] - \Tcal[u_2] = (u_1 - u_2) - (K[u_1] - K[u_2])$, the triangle inequality and Lemma~\ref{lem:KLip} give (\ref{eq:TLip}).
\end{proof}

\begin{corollary}[State norm bound from target norm]\label{cor:state_from_w}
Let $r_0 \in (0, r^*)$, so that $\Lipk(r_0) < 1$. For every $w \in \Bball_{\rho_w(r_0)}$, the corresponding $u = \Tcal^{-1}[w]$ produced by Lemma~\ref{lem:Tinv} of Appendix~\ref{app:Tinv} satisfies
\begin{equation}
\|u\|_\infty \;\le\; \LipT(r_0)\,\|w\|_\infty,
\label{eq:state_from_w}
\end{equation}
with $\LipT(r_0)$ given by~(\ref{eq:LipTinvdef}).
\end{corollary}

\begin{proof}
Apply Lemma~\ref{lem:TinvLip} of Appendix~\ref{app:Tinv} with $w_2 = 0$. Since $K[0] = 0$, the constant function $u \equiv 0$ satisfies $0 = 0 + K[0]$, so it is the unique fixed point in $\Bball_{r_0}$ given by Lemma~\ref{lem:Tinv}; that is, $\Tcal^{-1}[0] = 0$. Hence $\|u\|_\infty = \|\Tcal^{-1}[w] - \Tcal^{-1}[0]\|_\infty \le \LipT(r_0)\,\|w\|_\infty$.
\end{proof}

\section{The transport target system}\label{sec:transport}

Consider the linear transport problem
\begin{align}
w_t(x,t) &= w_x(x,t), \quad x \in [0,1),\ t \ge 0, \label{eq:transport_pde}\\
w(1,t) &= b(t), \label{eq:transport_bdy}\\
w(x,0) &= w_0(x). \label{eq:transport_ic}
\end{align}
The unique mild solution is given by the method of characteristics: for $(x,t) \in [0,1] \times [0,\infty)$,
\begin{equation}
w(x,t) \;=\; \begin{cases} w_0(x + t), & x + t \le 1,\\ b\bigl(t - (1-x)\bigr), & x + t > 1. \end{cases}
\label{eq:transport_solution}
\end{equation}

The boundary residual is the only source of energy after one transport-crossing time: the initial profile shifts to the right and exits the domain at $t = 1$, leaving the target state determined entirely by the recent history of the residual. The decay of the truncation tail therefore acts directly on the target.

\begin{lemma}[Sup-norm propagation under transport]\label{lem:transport}
Let $w_0 \in \Linf(0,1)$ and $b \in \Linf(0,\infty)$. The mild solution (\ref{eq:transport_solution}) of (\ref{eq:transport_pde}) satisfies, for every $t \ge 0$,
\begin{eqnarray}
\|w(\cdot,t)\|_\infty &\le& \max\bigl(\|w_0\|_\infty,\, \sup_{0\le \tau \le t} |b(\tau)|\bigr), \nonumber \\
&& \quad 0 \le t \le 1, \nonumber \\
\|w(\cdot,t)\|_\infty &\le& \sup_{t-1 \le \tau \le t}|b(\tau)|, \quad t \ge 1.
\label{eq:transport_bound}
\end{eqnarray}
\end{lemma}

\begin{proof}
Fix $t \ge 0$. From (\ref{eq:transport_solution}), the values $\{w(x,t) : x \in [0,1]\}$ are partitioned into two sets:
\begin{itemize}
\item $\{w_0(x+t) : x \in [0, \min(1, 1-t)]\}$, present only when $t < 1$, contributing at most $\|w_0\|_\infty$.
\item $\{b(t - (1-x)) : x \in [\max(0, 1-t), 1]\}$, equivalently $\{b(\tau) : \tau \in [\max(0, t-1), t]\}$, contributing at most $\sup_{\max(0,t-1) \le \tau \le t} |b(\tau)|$.
\end{itemize}
For $0 \le t \le 1$ both sets are nonempty, and $\|w(\cdot,t)\|_\infty$ is bounded by their maximum; for $t \ge 1$ only the second is present, the initial condition having been swept out of $[0,1]$.
\end{proof}

\section{Closed-loop well-posedness}\label{sec:wellposed}

\begin{definition}[Canonical mild solution]\label{rem:canonical_rep}
A \emph{canonical mild solution} of the closed-loop system (\ref{eq:plant})--(\ref{eq:Fdef}) under the truncated feedback $U(t) = K_N[u(\cdot, t)](1)$ is a measurable function $u: [0,1] \times [0, \infty) \to \R$ with $u \in \Linf_{\mathrm{loc}}([0, \infty); \Linf(0,1))$ satisfying, pointwise in $(x, t)$, for $x + t \le 1$,
\begin{eqnarray}
u(x, t) \;=\; u_0(x + t) + \int_0^t F[u(\cdot, s)](x + t - s)\, ds, \nonumber \\
\label{eq:canonical_mild_a}
\end{eqnarray}
and, for $x + t > 1$,
\begin{eqnarray}
u(x, t) &=& K_N[u(\cdot, t - (1-x))](1) \nonumber \\
&& + \int_{t - (1-x)}^t F[u(\cdot, s)](x + t - s)\, ds,
\label{eq:canonical_mild_b}
\end{eqnarray}
where the integrals are understood pointwise in $x$, with integrands bounded and measurable in $s$. All pointwise-in-$t$ statements in the sequel refer to this representative.
\end{definition}

For the well-posedness analysis, we use the sup-norm bound and Lipschitz constant of the plant nonlinearity (\ref{eq:Fdef}) on the closed ball $\overline B_R$ for $R \in [0, \rho_f)$:
\begin{align}
M_F(R) &:= D_f\,\frac{R^2}{\rho_f - R}, \label{eq:MFdef}\\
L_F(R) &:= D_f\!\left(\frac{1}{(1 - R/\rho_f)^2} - 1\right), \label{eq:LFdef}
\end{align}
together with the analogous quantities for the truncated controller $K_N$ on $\overline B_R$ for $R \in [0, 1/C_K)$:
\begin{align}
B_{K_N}(R) &:= D_K\, e^{\Upsilon_K}\, R\, \sum_{n=2}^N (C_K R)^{n-1}, \label{eq:BKNdef}\\
\LipKn(R) &:= D_K\, e^{\Upsilon_K}\, \sum_{n=2}^N n\,(C_K R)^{n-1}. \label{eq:LipKNdef}
\end{align}
The factor $n$ in $\LipKn$ arises from differentiating $u^{\otimes n}$ and is absent from $B_{K_N}$. Both $B_{K_N}$ and $\LipKn$ are partial sums of geometric series, with $B_{K_N}(R) \le R\, \Lipk(R)$ and $\LipKn(R) \le \Lipk(R)$ for every $R$ and $N$. That these four quantities serve as the claimed sup-norm bounds and Lipschitz constants of $F$ and $K_N$ is established in Appendix~\ref{app:Fbounds} (Lemmas~\ref{lem:Fbounds} and~\ref{lem:KNbounds}); the derivations are direct geometric summations from Assumption~\ref{ass:fnbound} and Proposition~\ref{prop:L2bound}.

\begin{lemma}[Local well-posedness of the truncated closed loop]\label{lem:wellposed}
Let Assumption~\ref{ass:fnbound} hold and fix $N \ge 2$ and $R \in (0, \bar r)$. For every $u_0 \in \Linf(0,1)$ with $\|u_0\|_\infty < R$, there exist $\tau > 0$ and a unique canonical mild solution $u \in \Linf([0, \tau]; \Linf(0,1))$ of the closed-loop system (\ref{eq:plant})--(\ref{eq:Fdef}) under the truncated feedback $U(t) = K_N[u(\cdot, t)](1)$ with $u(\cdot, 0) = u_0$ and $\|u(\cdot, t)\|_\infty \le R$ for every $t \in [0, \tau]$. Moreover, $\tau$ admits the lower bound
\begin{eqnarray}
\tau &\ge& \min\!\left(\frac{R - \max(\|u_0\|_\infty,\, B_{K_N}(R))}{M_F(R)},\right. \nonumber \\
&& \quad\quad \left.\frac{1 - \LipKn(R)}{2\, L_F(R)}\right),
\label{eq:tau_lb}
\end{eqnarray}
depending only on $R$, $\|u_0\|_\infty$, and the constants of Assumption~\ref{ass:fnbound} and Proposition~\ref{prop:L2bound}.
\end{lemma}

\begin{proof}
The standing condition $R < r^*$ gives $\LipKn(R) \le \Lipk(R) < 1$. Define the closed ball
\begin{eqnarray}
X_{\tau, R} &:=& \bigl\{u \in \Linf([0,\tau]; \Linf(0,1)) \,:\, \nonumber \\
&& \quad \|u(\cdot, t)\|_\infty \le R \text{ for every } t \in [0,\tau]\bigr\},
\label{eq:Xtau}
\end{eqnarray}
a closed (hence complete) subset of $\Linf([0,\tau]; \Linf(0,1))$ under the norm $\|u\|_{X_{\tau,R}} := \mathrm{ess\,sup}_{t \in [0,\tau]}\|u(\cdot, t)\|_\infty$. For $u \in X_{\tau, R}$, define the map $\Gamma u$ by the right-hand side of (\ref{eq:canonical_mild_a}--\ref{eq:canonical_mild_b}); the integrand $(x, s) \mapsto F[u(\cdot, s)](x)$ is bounded by $M_F(R)$ (Lemma~\ref{lem:Fbounds}) and measurable in $s$ for each $x$, so the integrals are well-defined as Lebesgue integrals in $s$.

\emph{Self-map.} For $u \in X_{\tau, R}$ and $(x, t)$ with $x + t \le 1$,
\begin{equation}
|(\Gamma u)(x, t)| \;\le\; \|u_0\|_\infty + \tau\, M_F(R).
\label{eq:Gamma_interior}
\end{equation}
For $x + t > 1$,
\begin{eqnarray}
|(\Gamma u)(x, t)| &\le& |K_N[u(\cdot, t - (1-x))](1)| + \tau\, M_F(R) \nonumber \\
&\le& B_{K_N}(R) + \tau\, M_F(R),
\label{eq:Gamma_boundary}
\end{eqnarray}
the last bound using $|K_N[v](1)| \le B_{K_N}(R)$ for $\|v\|_\infty \le R$ (the partial-sum bound in (\ref{eq:BKNdef})) and $\|u(\cdot, t-(1-x))\|_\infty \le R$. Hence
\begin{equation}
\|\Gamma u\|_{X_{\tau, R}} \;\le\; \max\bigl(\|u_0\|_\infty,\, B_{K_N}(R)\bigr) + \tau\, M_F(R).
\label{eq:Gamma_selfmap_bound}
\end{equation}
Choosing $\tau \le (R - \max(\|u_0\|_\infty, B_{K_N}(R)))/M_F(R)$ (omitted if $M_F(R) = 0$) gives $\|\Gamma u\|_{X_{\tau, R}} \le R$, so $\Gamma$ maps $X_{\tau, R}$ into itself. The numerator is strictly positive: $\|u_0\|_\infty < R$ by hypothesis, and $B_{K_N}(R) \le R\, \Lipk(R) < R$ since $\Lipk(R) < 1$.

\emph{Contraction.} For $u, v \in X_{\tau, R}$, the initial-data terms in (\ref{eq:canonical_mild_a}--\ref{eq:canonical_mild_b}) cancel. For $x + t \le 1$,
\begin{eqnarray}
|(\Gamma u - \Gamma v)(x, t)| &\le& \int_0^t \|F[u(\cdot, s)] - F[v(\cdot, s)]\|_\infty\, ds \nonumber \\
&\le& \tau\, L_F(R)\, \|u - v\|_{X_{\tau, R}},
\label{eq:Gamma_lip_interior}
\end{eqnarray}
the bound on $\|F[u]-F[v]\|_\infty$ following from Lemma~\ref{lem:Fbounds} of Appendix~\ref{app:Fbounds}. For $x + t > 1$,
\begin{eqnarray}
|(\Gamma u - \Gamma v)(x, t)| &\le& \LipKn(R)\,\|u - v\|_{X_{\tau, R}} \nonumber \\
&& + \tau\, L_F(R)\,\|u - v\|_{X_{\tau, R}}.
\label{eq:Gamma_lip_boundary}
\end{eqnarray}
Hence
\begin{eqnarray}
\|\Gamma u - \Gamma v\|_{X_{\tau, R}} &\le& \bigl(\LipKn(R) + \tau\, L_F(R)\bigr) \nonumber \\
&& \times \|u - v\|_{X_{\tau, R}}.
\label{eq:Gamma_contraction}
\end{eqnarray}
Choosing $\tau \le (1 - \LipKn(R))/(2\, L_F(R))$ (omitted if $L_F(R) = 0$) gives $\LipKn(R) + \tau\, L_F(R) \le (1 + \LipKn(R))/2 < 1$, so $\Gamma$ is a contraction on $X_{\tau, R}$. Banach's fixed-point theorem produces a unique fixed point $u \in X_{\tau, R}$ of $\Gamma$, which is the unique canonical mild solution of the closed-loop system on $[0, \tau]$.

The lower bound (\ref{eq:tau_lb}) follows from the two conditions on $\tau$ established above.

The fixed-point construction is invariant under time translation: for any $t_0 \ge 0$, if $u$ is a canonical mild solution on $[0, t_0 + \tau]$, the shift $\tilde u(\cdot, s) := u(\cdot, t_0 + s)$ satisfies the same fixed-point equation on $[0, \tau]$ with datum $\tilde u(\cdot, 0) = u(\cdot, t_0)$. Hence, by uniqueness in $X_{\tau, R}$, two solutions with values in $\overline B_R$ coinciding at some $t_0$ coincide on their common forward interval.
\end{proof}

By Lemma~\ref{lem:wellposed} and the time-shift invariance just stated, the local solutions can be extended by concatenation to a maximal canonical mild solution $u$ defined on an interval $[0, T_{\max})$, with the blow-up alternative
\begin{equation}
T_{\max} = \infty \quad \text{or} \quad \limsup_{t \uparrow T_{\max}} \|u(\cdot, t)\|_\infty = \infty.
\label{eq:blowup_alternative}
\end{equation}
Forward invariance is then a question about which branch of the alternative the trajectory takes. If the trajectory can be shown to stay strictly inside $\overline B_R$ on its entire interval of existence, with margin to spare, then the local existence time of Lemma~\ref{lem:wellposed} can be applied at each restart with a uniform lower bound, and concatenation forces the maximal interval to be infinite. The next lemma packages this contradiction.

\begin{lemma}[Continuation inside a ball]\label{lem:continuation}
Let Assumption~\ref{ass:fnbound} hold, $N \ge 2$, $R \in (0, \bar r)$, and $u_0 \in \Linf(0,1)$ with $\|u_0\|_\infty < R$. Let $u$ be the maximal canonical mild solution from Lemma~\ref{lem:wellposed} on $[0, T_{\max})$. If there exists $q < R$ such that
\begin{equation}
\|u(\cdot, t)\|_\infty \;\le\; q, \qquad \text{for every } t \in [0, T_{\max}),
\label{eq:cont_hyp}
\end{equation}
then $T_{\max} = \infty$.
\end{lemma}

\begin{proof}
Suppose for contradiction $T_{\max} < \infty$. The local existence of Lemma~\ref{lem:wellposed} can be applied with starting time $t_0 \in [0, T_{\max})$ and starting datum $u(\cdot, t_0)$, yielding a local solution on $[t_0, t_0 + \tau_*(t_0)]$ with values in $\overline B_R$, where $\tau_*(t_0)$ is given by (\ref{eq:tau_lb}) with $\|u_0\|_\infty$ replaced by $\|u(\cdot, t_0)\|_\infty$. The hypothesis (\ref{eq:cont_hyp}) gives $\|u(\cdot, t_0)\|_\infty \le q < R$, so
\begin{eqnarray}
\tau_*(t_0) &\ge& \tau_* \nonumber \\
&:=& \min\!\left(\frac{R - \max(q,\, B_{K_N}(R))}{M_F(R)},\right. \nonumber \\
&& \quad\quad \left.\frac{1 - \LipKn(R)}{2\, L_F(R)}\right) \;>\; 0,
\label{eq:tau_uniform_lb}
\end{eqnarray}
uniformly in $t_0 \in [0, T_{\max})$. By the contraction uniqueness in $X_{\tau_*, R}$ (Lemma~\ref{lem:wellposed}), the local solution from $t_0$ agrees on $[t_0, t_0 + \tau_*]$ with the maximal solution wherever both are defined. Choose $t_0 > T_{\max} - \tau_*/2$. The local solution from $t_0$ extends the maximal solution to $[0, t_0 + \tau_*]$, with $t_0 + \tau_* > T_{\max} + \tau_*/2 > T_{\max}$. This contradicts the maximality of $T_{\max}$.
\end{proof}

\section{Proof of the main theorem}\label{sec:mainproof}

We prove Theorem~\ref{thm:main} by combining the lemmas of Sections~\ref{sec:tail}--\ref{sec:transport}.

\paragraph*{\normalfont\em Target system under truncated feedback.}
Apply the backstepping transformation $w = \Tcal[u] = u - K[u]$ to the closed-loop plant (\ref{eq:plant})--(\ref{eq:Fdef}) under the truncated feedback $U(t) = K_N[u(\cdot,t)](1)$ from (\ref{eq:KNdef}). Following the construction in \cite{krstic2026feedbacklinearizationhyperbolicpdes}, the interior equation $w_t = w_x$ holds, and at the boundary $x = 1$,
\begin{eqnarray}
w(1,t) &=& u(1,t) - K[u(\cdot,t)](1) \nonumber \\
&=& K_N[u(\cdot,t)](1) - K[u(\cdot,t)](1).
\label{eq:wbdy}
\end{eqnarray}
Define the boundary residual
\begin{eqnarray}
b(t) &:=& w(1,t) \nonumber \\
&=& K_N[u(\cdot,t)](1) - K[u(\cdot,t)](1).
\label{eq:bdef}
\end{eqnarray}
Thus $w$ satisfies the linear transport equation (\ref{eq:transport_pde}) with boundary data $b(t)$ and initial condition $w_0 = u_0 - K[u_0] = \Tcal[u_0]$.

\paragraph*{\normalfont\em Existence of the closed-loop trajectory.}
By Lemma~\ref{lem:wellposed} applied with $R = r_0$ (using $\|u_0\|_\infty < \rho_u^0(r_0) < r_0$, which follows from (\ref{eq:cond_init})), the closed-loop system admits a unique maximal canonical mild solution $u$ on $[0, T_{\max})$ with values in $\overline B_{r_0}$. We show below that $T_{\max} = \infty$ (Step~1) and that the bounds (i)--(iv) hold (Steps~2--5).

\subsection*{A priori bound on $\|w_0\|_\infty$}

By (\ref{eq:cond_init}) and Lemma~\ref{lem:TLip} (applied with $u_2 = 0$, $u_1 = u_0$, $r = \|u_0\|_\infty$, recalling $K[0] = 0$ so $\Tcal[0] = 0$),
\begin{eqnarray}
\|w_0\|_\infty &=& \|\Tcal[u_0] - \Tcal[0]\|_\infty \nonumber \\
&\le& \bigl(1 + \Lipk(\|u_0\|_\infty)\bigr)\,\|u_0\|_\infty \nonumber \\
&\le& \bigl(1 + \Lipk(r_0)\bigr)\,\|u_0\|_\infty,
\label{eq:w0_bound}
\end{eqnarray}
the last inequality using monotonicity of $\Lipk$ and $\|u_0\|_\infty < \rho_u^0(r_0) < r_0$. Combined with (\ref{eq:cond_init}),
\begin{equation}
\|w_0\|_\infty \;<\; \bigl(1 + \Lipk(r_0)\bigr)\,\rho_u^0(r_0) \;=\; \rho_w(r_0),
\label{eq:w0_in_ball}
\end{equation}
the last equality using definitions (\ref{eq:rhowdef}) and (\ref{eq:rhou0def}).

\subsection*{Step 1: Forward invariance.}

\emph{Preliminary identity.} The definitions (\ref{eq:rhowdef}) and (\ref{eq:LipTinvdef}) imply
\begin{eqnarray}
\LipT(r_0)\,\rho_w(r_0) &=& \frac{1}{1 - \Lipk(r_0)}\cdot r_0\bigl(1 - \Lipk(r_0)\bigr) \nonumber \\
&=& r_0,
\label{eq:LipTrho_id}
\end{eqnarray}
which is invoked in subsequent steps.

On the maximal interval $[0, T_{\max})$, the trajectory takes values in $\overline B_{r_0}$ by Lemma~\ref{lem:wellposed}. Lemma~\ref{lem:trunc_tail} applied with $r = r_0$ gives
\begin{eqnarray}
|b(t)| &=& |K_N[u(\cdot,t)](1) - K[u(\cdot,t)](1)| \nonumber \\
&\le& \eps_N(r_0), \qquad t \in [0, T_{\max}),
\label{eq:bbound_invariant}
\end{eqnarray}
and Lemma~\ref{lem:transport} (applied to the target state $w := \Tcal[u]$) gives
\begin{eqnarray}
\|w(\cdot, t)\|_\infty &\le& \max\bigl(\|w_0\|_\infty,\, \eps_N(r_0)\bigr), \nonumber \\
&& \quad t \in [0, T_{\max}).
\label{eq:wnorm_invariant}
\end{eqnarray}
Both arguments of the maximum are strictly less than $\rho_w(r_0)$: $\eps_N(r_0) < \rho_w(r_0)$ by (\ref{eq:cond_resid}), and (\ref{eq:w0_in_ball}) gives $\|w_0\|_\infty < \rho_w(r_0)$. Hence $w(\cdot, t) \in \Bball_{\rho_w(r_0)}$ on $[0, T_{\max})$, and Corollary~\ref{cor:state_from_w} gives
\begin{eqnarray}
\|u(\cdot, t)\|_\infty &\le& \LipT(r_0)\,\|w(\cdot, t)\|_\infty \nonumber \\
&\le& q, \qquad t \in [0, T_{\max}),
\label{eq:unorm_strict}
\end{eqnarray}
where
\begin{eqnarray}
q &:=& \LipT(r_0)\,\max\bigl(\|w_0\|_\infty,\, \eps_N(r_0)\bigr) \nonumber \\
&<& \LipT(r_0)\,\rho_w(r_0) \;=\; r_0,
\label{eq:qdef}
\end{eqnarray}
the last equality being (\ref{eq:LipTrho_id}). The trajectory therefore stays in $\overline B_q$ with $q < r_0$ on $[0, T_{\max})$, with margin $r_0 - q > 0$. By Lemma~\ref{lem:continuation}, $T_{\max} = \infty$, proving (i).

\subsection*{Step 2: Practical stability.}

For every $t \ge 0$, since $\|u(\cdot,\tau)\|_\infty \le r_0$ for $\tau \in [0,t]$ by (i), Lemma~\ref{lem:trunc_tail} gives $|b(\tau)| \le \eps_N(r_0)$ on $[0,t]$, and Lemma~\ref{lem:transport} gives
\begin{equation}
\|w(\cdot, t)\|_\infty \;\le\; \max\bigl(\|w_0\|_\infty,\, \eps_N(r_0)\bigr), \quad t \ge 0,
\label{eq:wnorm_allt}
\end{equation}
where for $t \ge 1$ the $\|w_0\|_\infty$ argument can be dropped, but its inclusion preserves the form of the bound for all $t \ge 0$. By (\ref{eq:w0_in_ball}) and (\ref{eq:cond_resid}), both $\|w_0\|_\infty < \rho_w(r_0)$ and $\eps_N(r_0) < \rho_w(r_0)$, so $w(\cdot, t) \in \Bball_{\rho_w(r_0)}$, and Corollary~\ref{cor:state_from_w} applies:
\begin{eqnarray}
\|u(\cdot, t)\|_\infty &\le& \LipT(r_0)\,\|w(\cdot, t)\|_\infty \nonumber \\
&\le& \LipT(r_0)\,\max\bigl(\|w_0\|_\infty,\, \eps_N(r_0)\bigr).
\label{eq:proof_ii_w}
\end{eqnarray}
Folding in (\ref{eq:w0_bound}) and using $\LipT(r_0) = 1/(1-\Lipk(r_0))$ from (\ref{eq:LipTinvdef}), we obtain $\LipT(r_0)\,\|w_0\|_\infty \le (1+\Lipk(r_0))/(1-\Lipk(r_0))\,\|u_0\|_\infty$, and hence
\begin{eqnarray}
\|u(\cdot, t)\|_\infty &\le& \max\!\left(\frac{1 + \Lipk(r_0)}{1 - \Lipk(r_0)}\,\|u_0\|_\infty,\right. \nonumber \\
&& \quad\quad\;\; \left.\LipT(r_0)\,\eps_N(r_0)\right), \quad t \ge 0,
\label{eq:proof_ii}
\end{eqnarray}
proving (ii).

\subsection*{Step 3: Practical finite-time attractivity.}

For $t \ge 1$, by (\ref{eq:transport_solution}), $w(x, t) = b(t - (1-x))$ for all $x \in [0,1]$, so the initial data $w_0$ has been swept out of $[0,1]$. By (i), $\|u(\cdot, \tau)\|_\infty \le r_0$ for all $\tau \ge 0$, so $|b(\tau)| \le \eps_N(r_0)$ for all $\tau \ge 0$ by Lemma~\ref{lem:trunc_tail}. Hence
\begin{equation}
\|w(\cdot, t)\|_\infty \;\le\; \eps_N(r_0), \quad t \ge 1.
\label{eq:wnorm_t1}
\end{equation}
Since $\eps_N(r_0) < \rho_w(r_0)$ by (\ref{eq:cond_resid}), Corollary~\ref{cor:state_from_w} gives
\begin{eqnarray}
\|u(\cdot, t)\|_\infty &\le& \LipT(r_0)\,\|w(\cdot, t)\|_\infty \nonumber \\
&\le& \LipT(r_0)\,\eps_N(r_0), \quad t \ge 1,
\label{eq:proof_iii}
\end{eqnarray}
proving (iii).

\subsection*{Step 4: Asymptotic stability ($\mathcal{KL}$ estimate).}

We sharpen the analysis of (ii) by exploiting the higher-order vanishing of the truncation residual along the trajectory. Let $\rho(t) := \|u(\cdot,t)\|_\infty$ and $s := \|u_0\|_\infty$, and define
\begin{eqnarray}
\Phi(r) &:=& \LipT(r_0)\,\eps_N(r) \nonumber \\
&=& \frac{D_K\,e^{\Upsilon_K}}{C_K\,(1 - \Lipk(r_0))}\cdot\frac{(C_K r)^{N+1}}{1 - C_K r},
\label{eq:Phidef}
\end{eqnarray}
for $r \in [0, 1/C_K)$.
The function $\Phi$ is continuous and strictly increasing on $[0, 1/C_K)$, satisfies $\Phi(0) = 0$ and $\Phi(r) = O(r^{N+1})$ as $r \to 0$, and is class $\mathcal{K}$ on $[0, r_0]$. The map $r \mapsto \Phi(r)/r$ is strictly increasing on $(0, 1/C_K)$, and by the residual condition (\ref{eq:cond_resid}),
\begin{equation}
\frac{\Phi(r_0)}{r_0} \;=\; \frac{\LipT(r_0)\,\eps_N(r_0)}{r_0} \;=\; \frac{\eps_N(r_0)}{\rho_w(r_0)} \;<\; 1.
\label{eq:Phi_over_r_at_r0}
\end{equation}
Strict monotonicity of $r \mapsto \Phi(r)/r$ then gives
\begin{equation}
\Phi(r) \;<\; r \qquad \text{for every } r \in (0, r_0].
\label{eq:Phisubunity}
\end{equation}

\emph{Global amplification bound.} Set $Y_0 := \sup_{t \ge 0}\rho(t)$.
By (i), $Y_0 \le r_0$. Lemma~\ref{lem:trunc_tail} applied with $r = Y_0$ gives
\begin{equation}
|b(t)| \;\le\; \eps_N(Y_0), \quad t \ge 0,
\label{eq:b_bounded_Y0}
\end{equation}
and Lemma~\ref{lem:transport} therefore yields
\begin{equation}
\|w(\cdot,t)\|_\infty \;\le\; \max\bigl(\|w_0\|_\infty,\, \eps_N(Y_0)\bigr), \quad t \ge 0.
\label{eq:w_norm_Y0}
\end{equation}
Both arguments of the maximum are strictly less than $\rho_w(r_0)$: $\|w_0\|_\infty < \rho_w(r_0)$ by (\ref{eq:w0_in_ball}), and $\eps_N(Y_0) \le \eps_N(r_0) < \rho_w(r_0)$ by monotonicity of $\eps_N$ and (\ref{eq:cond_resid}). Corollary~\ref{cor:state_from_w} then gives
\begin{eqnarray}
\rho(t) &\le& \LipT(r_0)\,\max\bigl(\|w_0\|_\infty,\, \eps_N(Y_0)\bigr) \nonumber \\
&=& \max\bigl(\LipT(r_0)\|w_0\|_\infty,\, \Phi(Y_0)\bigr), \quad t \ge 0.
\label{eq:rho_max}
\end{eqnarray}
Taking the supremum over $t \ge 0$,
\begin{equation}
Y_0 \;\le\; \max\bigl(\LipT(r_0)\|w_0\|_\infty,\, \Phi(Y_0)\bigr).
\label{eq:Y0_max}
\end{equation}
If the maximum in (\ref{eq:Y0_max}) is attained by the second term and $Y_0 > 0$, then $Y_0 \le \Phi(Y_0) < Y_0$ by (\ref{eq:Phisubunity}), a contradiction. Hence
\begin{equation}
Y_0 \;\le\; \LipT(r_0)\,\|w_0\|_\infty.
\label{eq:Y0_linear}
\end{equation}
Combining with (\ref{eq:w0_bound}) and using $\LipT(r_0) = 1/(1 - \Lipk(r_0))$,
\begin{equation}
Y_0 \;\le\; \alpha(s), \qquad \alpha(s) := \frac{1 + \Lipk(r_0)}{1 - \Lipk(r_0)}\, s.
\label{eq:Y0_alpha}
\end{equation}
The function $\alpha$ is class $\mathcal{K}$, and by (\ref{eq:cond_init}) and (\ref{eq:LipTrho_id}), $\alpha(s) < r_0$ for every $s \in [0, \rho_u^0(r_0))$.

\emph{Tail-supremum recursion.} For $k = 0, 1, 2, \ldots,$ define $Y_k := \sup_{t \ge k}\rho(t)$.
For $t \ge k+1$, the inclusion $[t-1, t] \subseteq [k, \infty)$ and monotonicity of $\eps_N$ give, using Lemma~\ref{lem:transport} on $[t-1, t]$,
\begin{eqnarray}
\|w(\cdot, t)\|_\infty &\le& \sup_{\tau \in [t-1, t]} |b(\tau)| \nonumber \\
&\le& \eps_N\!\left(\sup_{\tau \in [t-1, t]}\rho(\tau)\right) \nonumber \\
&\le& \eps_N(Y_k).
\label{eq:w_in_window}
\end{eqnarray}
Since $Y_k \le Y_0 \le r_0$ and $\eps_N(Y_k) \le \eps_N(r_0) < \rho_w(r_0)$, Corollary~\ref{cor:state_from_w} applies and gives
\begin{equation}
\rho(t) \;\le\; \LipT(r_0)\,\eps_N(Y_k) \;=\; \Phi(Y_k), \qquad t \ge k+1.
\label{eq:rho_Phi_Yk}
\end{equation}
Taking the supremum over $t \ge k+1$,
\begin{equation}
Y_{k+1} \;\le\; \Phi(Y_k), \qquad k = 0, 1, 2, \ldots.
\label{eq:Yk_recursion}
\end{equation}

\emph{Iterates.} Define
\begin{eqnarray}
\sigma_0(s) &:=& \alpha(s), \nonumber \\
\sigma_{k+1}(s) &:=& \Phi\bigl(\sigma_k(s)\bigr),
\label{eq:sigma_iterate}
\end{eqnarray}
for $k \ge 0$ and $s \in [0, \rho_u^0(r_0))$.
By (\ref{eq:Y0_alpha}), $Y_0 \le \sigma_0(s)$. Inductive use of (\ref{eq:Yk_recursion}) and the monotonicity of $\Phi$ gives
\begin{equation}
Y_k \;\le\; \sigma_k(s), \qquad k = 0, 1, 2, \ldots.
\label{eq:Yk_sigma}
\end{equation}
For $s \in (0, \rho_u^0(r_0))$, $\sigma_0(s) < r_0$, so (\ref{eq:Phisubunity}) gives $\sigma_{k+1}(s) = \Phi(\sigma_k(s)) < \sigma_k(s)$ as long as $\sigma_k(s) > 0$. The sequence $(\sigma_k(s))_{k \ge 0}$ is therefore strictly decreasing in $k$, and its limit $\sigma_\infty(s)$ satisfies $\sigma_\infty(s) = \Phi(\sigma_\infty(s))$ by continuity. By (\ref{eq:Phisubunity}), $\Phi$ has only the fixed point $0$ in $[0, r_0]$, so $\sigma_\infty(s) = 0$. Each $\sigma_k$ is a composition of class-$\mathcal{K}$ functions and is therefore class $\mathcal{K}$ in $s$.

\emph{Class-$\mathcal{KL}$ envelope.} The recursion gives a piecewise-constant majorant $\rho(t) \le \sigma_{\lfloor t \rfloor}(s)$, non-increasing but discontinuous in $t$. To produce a continuous strictly-decreasing class-$\mathcal{KL}$ envelope, set $\sigma_{-1}(s) := 2\sigma_0(s)$ and define
\begin{eqnarray}
\beta(s, t) &:=& (1-\theta)\,\sigma_{k-1}(s) + \theta\,\sigma_k(s), \nonumber \\
&& \quad \theta := t - k, \ t \in [k, k+1], \ k \ge 0,
\label{eq:beta_def}
\end{eqnarray}
for $s \in [0, \rho_u^0(r_0))$. The construction interpolates linearly on each unit interval between $\sigma_{k-1}(s)$ at $t = k$ and $\sigma_k(s)$ at $t = k+1$, with the choice $\sigma_{-1} := 2\sigma_0$ ensuring strict decrease on the initial interval $[0, 1]$.

\emph{Verification of class-$\mathcal{KL}$.} For each fixed $t \ge 0$, $\beta(\cdot, t)$ is a positive convex combination of class-$\mathcal{K}$ functions of $s$ (with coefficients $1-\theta, \theta \ge 0$ summing to $1$) and is therefore class $\mathcal{K}$. For each fixed $s \in (0, \rho_u^0(r_0))$, $\beta(s, \cdot)$ is continuous (the linear pieces meet at integer endpoints with matching values $\sigma_k(s)$) and strictly decreasing in $t$ (the endpoints satisfy $\sigma_{k-1}(s) > \sigma_k(s)$ for every $k \ge 0$, by the strict decrease of $(\sigma_k)$ established above); on each $[k, k+1]$ the linear interpolant inherits this strict decrease. Finally, $\beta(s, t) \to 0$ as $t \to \infty$, since $\sigma_k(s) \to 0$. Hence $\beta$ is class $\mathcal{KL}$.

\emph{Bound.} For $t \in [k, k+1]$ with $k \ge 0$, the linear interpolant satisfies $\beta(s, t) \ge \sigma_k(s)$ (it lies above its smaller endpoint), hence
\begin{equation}
\rho(t) \;\le\; Y_k \;\le\; \sigma_k(s) \;\le\; \beta(s, t),
\label{eq:rho_under_beta}
\end{equation}
giving
\begin{equation}
\|u(\cdot, t)\|_\infty \;\le\; \beta\bigl(\|u_0\|_\infty,\, t\bigr), \quad t \ge 0,
\label{eq:KL_assembled}
\end{equation}
proving (iv).

\subsection*{Step 5: Geometric decay of the residual in $N$.}

From (\ref{eq:epsNdef}), $\eps_N(r_0) = \frac{D_K e^{\Upsilon_K}}{C_K}\cdot \frac{(C_K r_0)^{N+1}}{1 - C_K r_0}$, geometric in $N$ with ratio $C_K r_0 < 1$. As $N \to \infty$, $\eps_N(r_0) \to 0$, so the residual $\LipT(r_0)\,\eps_N(r_0)$ in (ii) and (iii) vanishes, and the closed-loop system recovers the exact-feedback behavior: under the exact (untruncated) feedback, $w \equiv 0$ for $t \ge 1$, and hence $u \equiv 0$ for $t \ge 1$ on the local inverse branch where $\Tcal^{-1}$ is defined. \qed

\section{Stability in $\Ltwo$ by truncated transformation}\label{sec:L2}

The main theorem is a sup-norm result, and the Conclusions of \cite{krstic2026feedbacklinearizationhyperbolicpdes} suggest that truncation forecloses the energy-norm setting native to backstepping. It does not. This section proves that the same truncated feedback also stabilizes the closed loop in $\Ltwo$ (Theorem~\ref{thm:L2}), so the choice of norm in Theorem~\ref{thm:main} is a matter of sharpness, not of necessity.

The route is to truncate the backstepping transformation along with the feedback. Set
\begin{equation}
\Tcal_N \;:=\; I - K_N, \qquad \widetilde K_N \;:=\; K - K_N,
\label{eq:TNdef}
\end{equation}
and $w_N := \Tcal_N[u] = u - K_N[u]$. The feedback is unchanged, $U(t) = K_N[u(\cdot,t)](1)$, so at the boundary
\begin{eqnarray}
w_N(1,t) &=& U(t) - K_N[u(\cdot,t)](1) \;=\; 0,
\label{eq:wN_bdy}
\end{eqnarray}
and the target's boundary condition closes exactly. This introduces no new dynamics. With $w := \Tcal[u]$ the exact-transformation target of Section~\ref{sec:mainproof} --- which satisfies the transport system (\ref{eq:transport_pde})--(\ref{eq:transport_ic}) with $b(t) = -\widetilde K_N[u(\cdot,t)](1)$, a system in no way tied to the sup-norm --- the identity $\Tcal_N = \Tcal + \widetilde K_N$ gives
\begin{equation}
w_N(x,t) \;=\; w(x,t) + \widetilde K_N[u(\cdot,t)](x),
\label{eq:wN_decomp}
\end{equation}
so the truncation mismatch is a static order-$(N{+}1)$ term added to the existing target through two evaluations of the tail: its trace $b(t)$, which drives $w$, and its profile $\widetilde K_N[u]$. Both, as shown next, are controlled by the $\Ltwo$ norm of the state; propagation is carried out on $w$ via (\ref{eq:transport_solution}), and the state is recovered by inverting the finite operator $\Tcal_N$.

The tail bound is the first ingredient, and it costs nothing to move to $\Ltwo$: the proof of Lemma~\ref{lem:pointwise_tail} already runs in $\|u\|_{\Ltwo}$ and slice-wise in $x$, relaxing to the sup-norm only in its last line.

\begin{lemma}[$\Ltwo$ tail-operator bound]\label{lem:tailL2}
Under Assumption~\ref{ass:fnbound} and Proposition~\ref{prop:L2bound}, for every $u \in \Ltwo(0,1)$ with $\|u\|_{\Ltwo} < 1/C_K$ and every $N \ge 2$,
\begin{equation}
\sup_{x \in [0,1]} \bigl|\widetilde K_N[u](x)\bigr| \;\le\; \eps_N\bigl(\|u\|_{\Ltwo}\bigr),
\label{eq:tailL2}
\end{equation}
with $\eps_N$ as in (\ref{eq:epsNdef}). In particular, both the trace and the profile of the tail obey
\begin{eqnarray}
\bigl|\widetilde K_N[u](1)\bigr| &\le& \eps_N\bigl(\|u\|_{\Ltwo}\bigr), \nonumber \\
\bigl\|\widetilde K_N[u]\bigr\|_{\Ltwo} &\le& \eps_N\bigl(\|u\|_{\Ltwo}\bigr).
\label{eq:tailL2_two}
\end{eqnarray}
\end{lemma}

\begin{proof}
Fix $x \in [0,1]$ and $n \ge N+1$. By Cauchy--Schwarz on $T_n(x)$ and the symmetrization identity (\ref{eq:symmetrization}) applied on $T_n(x)$ --- available because $|u^{\otimes n}|^2$ is a symmetric integrand --- $\|u^{\otimes n}\|_{\Ltwo(T_n(x))} = \|u\|_{\Ltwo(0,x)}^n/\sqrt{n!}$. Pairing with the slice bound (\ref{eq:knL2}), the factor $\sqrt{n!}$ cancels:
\begin{equation}
\biggl|\int_{T_n(x)} k_n(x,\cdot)\, u^{\otimes n}\, d\xi\biggr| \;\le\; \frac{D_K\, e^{\Upsilon_K}}{C_K}\,\bigl(C_K\,\|u\|_{\Ltwo}\bigr)^n,
\label{eq:tailL2_pointwise}
\end{equation}
uniformly in $x$. Summing the geometric series from $n = N+1$ gives (\ref{eq:tailL2}); the consequences (\ref{eq:tailL2_two}) follow since $[0,1]$ has unit Lebesgue measure.
\end{proof}

The strengthening over Lemma~\ref{lem:trunc_tail} --- an $\Ltwo$ argument, and control at every $x$ rather than only the trace --- is free: the integrand $|u^{\otimes n}|^2$ is permutation-invariant, so the full simplex gain $1/n!$ cancels the factorial of (\ref{eq:knL2}), as in Section~\ref{sec:KLip}. The same symmetry gives, for $\|u\|_{\Ltwo} \le R$,
\begin{eqnarray}
\bigl|K_N[u](1)\bigr| &\le& B_{K_N}(R), \qquad R < 1/C_K, \nonumber \\
\|F[u]\|_\infty &\le& M_F(R), \qquad\ \ R < \rho_f,
\label{eq:BKN_MF_L2}
\end{eqnarray}
with $B_{K_N}$, $M_F$ as in (\ref{eq:BKNdef}), (\ref{eq:MFdef}), the plant bound using $\|u\|_{L^1} \le \|u\|_{\Ltwo}$.

The Lipschitz estimates are where $\Ltwo$ exacts a price. The telescoped difference (\ref{eq:telescope}) places a group of $u_1$ factors, the difference $u_1 - u_2$, and a group of $u_2$ factors in the slots, so the squared integrand is only partly symmetric: symmetrizing within each group --- a \emph{partial} symmetrization (Appendix~\ref{app:L2}) --- recovers $1/((j-1)!\,(n-j)!)$ rather than $1/n!$, and the uncancelled interleavings, entering under a square root, renormalize the constants to $C_K \mapsto \sqrt{2}\,C_K$ for the controller and $\rho_f \mapsto \rho_f/2$ for the plant.

\begin{lemma}[$\Ltwo$ Lipschitzness of $K$ and $K_N$]\label{lem:KLipL2}
Under Assumption~\ref{ass:fnbound} and Proposition~\ref{prop:L2bound}, for every $r \in [0, 1/(\sqrt{2}\,C_K))$ and every $u_1, u_2 \in \Ltwo(0,1)$ with $\|u_1\|_{\Ltwo}, \|u_2\|_{\Ltwo} \le r$,
\begin{equation}
\bigl\|K[u_1] - K[u_2]\bigr\|_{\Ltwo} \;\le\; \Lipk^{(2)}(r)\,\|u_1 - u_2\|_{\Ltwo},
\label{eq:KLipL2}
\end{equation}
where
\begin{equation}
\Lipk^{(2)}(r) \;:=\; D_K\, e^{\Upsilon_K}\, \frac{\sqrt{2}\,C_K r\,\bigl(2 - \sqrt{2}\,C_K r\bigr)}{\bigl(1 - \sqrt{2}\,C_K r\bigr)^2},
\label{eq:LipK2def}
\end{equation}
i.e., (\ref{eq:LipKdef}) with $C_K$ replaced by $\sqrt{2}\,C_K$. The same bound with the sum stopped at $n = N$ holds for $K_N$ with the constant
\begin{eqnarray}
\LipKn^{(2)}(r) &:=& D_K\, e^{\Upsilon_K}\, \sum_{n=2}^N n\,\bigl(\sqrt{2}\,C_K r\bigr)^{n-1} \nonumber \\
&\le& \Lipk^{(2)}(r).
\label{eq:LipKN2def}
\end{eqnarray}
\end{lemma}

The $K$ case of (\ref{eq:KLipL2}) is \cite[Lemma~2]{krstic2026feedbacklinearizationhyperbolicpdes}, re-derived in the $(C_K, D_K, \Upsilon_K)$ currency of Part~I and extended to $K_N$ with a constant uniform in $N$.

\begin{proof} See Appendix~\ref{app:L2}. \end{proof}

The function $\Lipk^{(2)}$ is continuous and strictly increasing on $[0, 1/(\sqrt{2}\,C_K))$, vanishes at $0$, and blows up at the right endpoint, so $\Lipk^{(2)}(r_2^*) = 1$ has a unique root $r_2^* \in (0, 1/(\sqrt{2}\,C_K))$; since $\Lipk(r) \le \Lipk^{(2)}(r)$ pointwise, $r_2^* \le r^*$. Set $\bar r_2 := \min(r_2^*,\, \rho_f/2)$, and, for $r_0 \in (0, \bar r_2)$, define the $\Ltwo$ analogues of (\ref{eq:rhowdef}) and (\ref{eq:LipTinvdef}),
\begin{eqnarray}
\rho_w^{(2)}(r_0) &:=& r_0\,\bigl(1 - \Lipk^{(2)}(r_0)\bigr), \label{eq:rhow2def}\\
\LipT^{(2)}(r_0) &:=& \frac{1}{1 - \Lipk^{(2)}(r_0)}. \label{eq:LipT2def}
\end{eqnarray}

The target-side estimate transfers to $\Ltwo$ with the sup-norm maximum replaced by a Pythagorean sum, the initial-data and boundary contributions occupying disjoint regions of the spatial domain.

\begin{lemma}[$\Ltwo$ propagation under transport]\label{lem:transportL2}
Let $w_0 \in \Ltwo(0,1)$ and $b \in \Linf(0,\infty)$. The mild solution (\ref{eq:transport_solution}) of (\ref{eq:transport_pde})--(\ref{eq:transport_ic}) satisfies
\begin{eqnarray}
\|w(\cdot,t)\|_{\Ltwo}^2 &\le& \|w_0\|_{\Ltwo}^2 + \sup_{0 \le \tau \le t}|b(\tau)|^2, \quad t \ge 0, \nonumber \\
\|w(\cdot,t)\|_{\Ltwo} &\le& \sup_{t-1 \le \tau \le t}|b(\tau)|, \qquad\quad\ \  t \ge 1.
\label{eq:transportL2}
\end{eqnarray}
\end{lemma}

\begin{proof}
From (\ref{eq:transport_solution}), after the substitution $\tau = t - (1-x)$ in the boundary part,
\begin{eqnarray}
\|w(\cdot,t)\|_{\Ltwo}^2 &=& \int_0^{\max(0,\,1-t)} |w_0(x+t)|^2\, dx \nonumber \\
&& +\; \int_{\max(0,\,t-1)}^{t} |b(\tau)|^2\, d\tau.
\label{eq:transportL2_split}
\end{eqnarray}
The first term is at most $\|w_0\|_{\Ltwo}^2$ and vanishes for $t \ge 1$; the second is at most $\min(t,1)\, \sup_{\max(0,t-1) \le \tau \le t}|b(\tau)|^2$.
\end{proof}

The pointwise structure of (\ref{eq:canonical_mild_a})--(\ref{eq:canonical_mild_b}) makes the well-posedness argument of Section~\ref{sec:wellposed} portable: every bound entering the fixed point is uniform in $x$, so the estimates survive intact when the state ball is measured in $\Ltwo$, with the Lipschitz constants replaced by their partially symmetrized $\Ltwo$ versions.

\begin{lemma}[$\Ltwo$ well-posedness of the truncated closed loop]\label{lem:wellposedL2}
Let Assumption~\ref{ass:fnbound} hold and fix $N \ge 2$ and $R \in (0, \bar r_2)$. For every $u_0 \in \Ltwo(0,1)$ with $\|u_0\|_{\Ltwo} < R$, there exist $\tau > 0$ and a unique canonical mild solution $u \in \Linf([0,\tau]; \Ltwo(0,1))$ of the closed loop (\ref{eq:plant})--(\ref{eq:Fdef}) under $U(t) = K_N[u(\cdot,t)](1)$ --- in the sense of Definition~\ref{rem:canonical_rep} with $\Ltwo(0,1)$ in place of $\Linf(0,1)$, the identities (\ref{eq:canonical_mild_a})--(\ref{eq:canonical_mild_b}) read a.e.\ in $x$ and $u_0(x+t)$ read as the $\Ltwo$ shift --- with $u(\cdot,0) = u_0$, $\|u(\cdot,t)\|_{\Ltwo} \le R$ on $[0,\tau]$, and
\begin{equation}
\tau \ge \min\!\Biggl(1,\; \biggl(\frac{R - \|u_0\|_{\Ltwo}}{B_{K_N}(R) + M_F(R)}\biggr)^{\!2}, \frac{1 - \LipKn^{(2)}(R)}{2\, L_F^{(2)}(R)}\Biggr).
\label{eq:tau_lb_L2}
\end{equation}
Moreover, the continuation criterion of Lemma~\ref{lem:continuation} holds with $\Ltwo$ in place of $\Linf$: if the maximal solution satisfies $\|u(\cdot,t)\|_{\Ltwo} \le q < R$ on $[0, T_{\max})$, then $T_{\max} = \infty$.
\end{lemma}

\begin{proof}
The fixed-point construction of Lemma~\ref{lem:wellposed} applies on
\begin{eqnarray}
X^{(2)}_{\tau,R} &:=& \bigl\{u \in \Linf([0,\tau]; \Ltwo(0,1)) \,:\, \nonumber \\
&& \quad \|u(\cdot,t)\|_{\Ltwo} \le R \text{ for every } t \in [0,\tau]\bigr\},
\label{eq:X2def}
\end{eqnarray}
with three changes. \emph{Self-map:} the non-integral part of $\Gamma u$ contributes $u_0(x+t)$ on $\{x + t \le 1\}$ and the trace value $K_N[u(\cdot, t-(1-x))](1)$, bounded by $B_{K_N}(R)$ via (\ref{eq:BKN_MF_L2}), on $\{x + t > 1\}$; the two sets are disjoint in $x$, the second having measure $\min(t,1)$, so for $t \le \tau \le 1$,
\begin{eqnarray}
\|\Gamma u(\cdot,t)\|_{\Ltwo} &\le& \bigl(\|u_0\|_{\Ltwo}^2 + t\, B_{K_N}(R)^2\bigr)^{1/2} \nonumber \\
&& +\; t\, M_F(R) \nonumber \\
&\le& \|u_0\|_{\Ltwo} + \sqrt{\tau}\,\bigl(B_{K_N}(R) + M_F(R)\bigr),
\label{eq:selfmapL2}
\end{eqnarray}
which does not exceed $R$ under the second entry of the minimum in (\ref{eq:tau_lb_L2}). \emph{Contraction:} the pointwise estimates (\ref{eq:Gamma_lip_interior})--(\ref{eq:Gamma_lip_boundary}) hold with $(\LipKn, L_F)$ replaced by $(\LipKn^{(2)}, L_F^{(2)})$ and the difference measured in the norm of $X^{(2)}_{\tau,R}$, by Lemmas~\ref{lem:KLipL2} and~\ref{lem:FLipL2}; the bounds are uniform in $x$, hence hold in $\Ltwo(0,1)$, and the contraction factor is at most $(1 + \LipKn^{(2)}(R))/2 < 1$ under the third entry of the minimum. \emph{Continuation:} the argument of Lemma~\ref{lem:continuation} applies with $\Ltwo$ restart datum $u(\cdot,t_0)$, the uniform restart time being (\ref{eq:tau_lb_L2}) with $\|u_0\|_{\Ltwo}$ replaced by $q$.
\end{proof}

With the tail bounds, the partially symmetrized Lipschitz constants, the inverse of $\Tcal_N$, transport propagation, and well-posedness all in hand, the closed-loop result follows by rerunning the five-step argument of Section~\ref{sec:mainproof} in the $\Ltwo$ norm.

\begin{theorem}[Closed-loop $\Ltwo$ stability under truncated feedback]\label{thm:L2}
Let Assumption~\ref{ass:fnbound} hold, and let $D_K, C_K, \Upsilon_K$ be the constants of Proposition~\ref{prop:L2bound}. Fix $r_0 \in (0, \bar r_2)$ and a truncation order $N \ge 2$ such that
\begin{equation}
2\,\eps_N(r_0) \;<\; \rho_w^{(2)}(r_0).
\label{eq:cond_residL2}
\end{equation}
Let $u_0 \in \Ltwo(0,1)$ satisfy
\begin{equation}
\|u_0\|_{\Ltwo} \;<\; \rho_u^{0,(2)}(r_0) \;:=\; \frac{\rho_w^{(2)}(r_0) - 2\,\eps_N(r_0)}{1 + \Lipk^{(2)}(r_0)},
\label{eq:cond_initL2}
\end{equation}
the right-hand side being positive by (\ref{eq:cond_residL2}). Then the closed-loop system (\ref{eq:plant})--(\ref{eq:Fdef}) under the truncated feedback $U(t) = K_N[u](1,t)$ from (\ref{eq:KNdef}) admits a unique canonical mild solution $u \in \Linf_{\mathrm{loc}}([0,\infty); \Ltwo(0,1))$ with $u(\cdot,0) = u_0$, satisfying:
\begin{enumerate}[label=(\roman*)]
\item \emph{(Forward invariance.)} $\|u(\cdot,t)\|_{\Ltwo} \le r_0$ for every $t \ge 0$.
\item \emph{(Practical stability.)} For all $t \ge 0$,
\begin{eqnarray}
\|u(\cdot,t)\|_{\Ltwo} &\le& \frac{1 + \Lipk^{(2)}(r_0)}{1 - \Lipk^{(2)}(r_0)}\,\|u_0\|_{\Ltwo} \nonumber \\
&& +\; 2\,\LipT^{(2)}(r_0)\,\eps_N(r_0).
\label{eq:thmL2_ii}
\end{eqnarray}
\item \emph{(Practical finite-time attractivity.)} For all $t \ge 1$,
\begin{equation}
\|u(\cdot,t)\|_{\Ltwo} \;\le\; 2\,\LipT^{(2)}(r_0)\,\eps_N(r_0).
\label{eq:thmL2_iii}
\end{equation}
\item \emph{(Asymptotic stability.)} There exists a function $\beta_2$ of class $\mathcal{KL}$ on $[0, \rho_u^{0,(2)}(r_0)) \times [0,\infty)$ such that
\begin{equation}
\|u(\cdot,t)\|_{\Ltwo} \;\le\; \beta_2\bigl(\|u_0\|_{\Ltwo},\, t\bigr), \qquad t \ge 0.
\label{eq:thmL2_iv}
\end{equation}
\end{enumerate}
The residual in (ii) and (iii) decays geometrically in the truncation order, $\eps_N(r_0) = O\bigl((C_K r_0)^N\bigr)$, exactly as in Theorem~\ref{thm:main}; as $N \to \infty$, the envelope (\ref{eq:thmL2_ii}) collapses to its linear term, the admissible ball (\ref{eq:cond_initL2}) expands to $r_0\,\bigl(1 - \Lipk^{(2)}(r_0)\bigr)/\bigl(1 + \Lipk^{(2)}(r_0)\bigr)$, and the exact-feedback $\Ltwo$ stabilization of \cite{krstic2026feedbacklinearizationhyperbolicpdes} is recovered, uniformly in time, on that ball.
\end{theorem}

\begin{proof}
\emph{Existence and target.} Lemma~\ref{lem:wellposedL2} with $R = r_0$ (admissible since $\|u_0\|_{\Ltwo} < \rho_u^{0,(2)}(r_0) < r_0$) produces the unique maximal canonical mild solution on $[0, T_{\max})$ with $\Ltwo$-values in the ball of radius $r_0$. Since $r_0 < 1/(\sqrt{2}\,C_K) < 1/C_K$, the exact transformation $w = \Tcal[u]$ is defined along the trajectory, and the mild-level identity of \cite{krstic2026feedbacklinearizationhyperbolicpdes} --- whose native setting is $\Ltwo$ --- gives that $w$ is the $\Ltwo$ mild solution of the transport equation (\ref{eq:transport_pde}) with boundary data $b(t) = K_N[u(\cdot,t)](1) - K[u(\cdot,t)](1) = -\widetilde K_N[u(\cdot,t)](1)$ and initial datum $w_0 = \Tcal[u_0]$, exactly as in Section~\ref{sec:mainproof}. The truncated-transformation state is then $w_N = w + \widetilde K_N[u]$ with $w_N(1,t) = 0$, per (\ref{eq:wN_bdy}) and (\ref{eq:wN_decomp}).

\emph{Step 1: forward invariance.} On $[0, T_{\max})$, Lemma~\ref{lem:tailL2} gives $|b(\tau)| \le \eps_N(r_0)$ and $\|\widetilde K_N[u(\cdot,t)]\|_{\Ltwo} \le \eps_N(r_0)$. Lemma~\ref{lem:transportL2} and the triangle inequality in (\ref{eq:wN_decomp}) then give
\begin{eqnarray}
\|w_N(\cdot,t)\|_{\Ltwo} &\le& \bigl(\|w_0\|_{\Ltwo}^2 + \eps_N(r_0)^2\bigr)^{1/2} + \eps_N(r_0) \nonumber \\
&\le& \|w_0\|_{\Ltwo} + 2\,\eps_N(r_0).
\label{eq:wN_bound_inv}
\end{eqnarray}
By Lemma~\ref{lem:KLipL2} with $u_2 = 0$ and $K[0] = 0$, $\|w_0\|_{\Ltwo} \le \bigl(1 + \Lipk^{(2)}(r_0)\bigr)\|u_0\|_{\Ltwo}$, so (\ref{eq:cond_initL2}) gives $\|w_0\|_{\Ltwo} < \rho_w^{(2)}(r_0) - 2\,\eps_N(r_0)$, and (\ref{eq:wN_bound_inv}) yields $\|w_N(\cdot,t)\|_{\Ltwo} < \rho_w^{(2)}(r_0)$ with a $t$-uniform margin. Since $\Tcal_N[u(\cdot,t)] = w_N(\cdot,t)$ with $\|u(\cdot,t)\|_{\Ltwo} \le r_0$, the uniqueness clause of Lemma~\ref{lem:TNinv} identifies $u(\cdot,t) = \Tcal_N^{-1}[w_N(\cdot,t)]$, whence
\begin{eqnarray}
\|u(\cdot,t)\|_{\Ltwo} &\le& \LipT^{(2)}(r_0)\,\bigl(\|w_0\|_{\Ltwo} + 2\,\eps_N(r_0)\bigr) \nonumber \\
&<& \LipT^{(2)}(r_0)\,\rho_w^{(2)}(r_0) \;=\; r_0,
\label{eq:u_margin_L2}
\end{eqnarray}
uniformly on $[0, T_{\max})$. The continuation criterion of Lemma~\ref{lem:wellposedL2} forces $T_{\max} = \infty$, proving (i).

\emph{Step 2: practical stability.} Folding $\|w_0\|_{\Ltwo} \le (1 + \Lipk^{(2)}(r_0))\|u_0\|_{\Ltwo}$ into the first line of (\ref{eq:u_margin_L2}) and using $\LipT^{(2)} = 1/(1 - \Lipk^{(2)})$ gives (\ref{eq:thmL2_ii}) for all $t \ge 0$.

\emph{Step 3: finite-time attractivity.} For $t \ge 1$ the initial datum has been swept out of $[0,1]$: the second line of (\ref{eq:transportL2}) gives $\|w(\cdot,t)\|_{\Ltwo} \le \sup_{t-1 \le \tau \le t}|b(\tau)| \le \eps_N(r_0)$, hence $\|w_N(\cdot,t)\|_{\Ltwo} \le 2\,\eps_N(r_0) < \rho_w^{(2)}(r_0)$, and the conversion through $\Tcal_N^{-1}$ gives (\ref{eq:thmL2_iii}).

\emph{Step 4: asymptotic stability.} Set $\rho_2(t) := \|u(\cdot,t)\|_{\Ltwo}$, $s := \|u_0\|_{\Ltwo}$, and
\begin{equation}
\Phi_2(r) \;:=\; 2\,\LipT^{(2)}(r_0)\,\eps_N(r), \qquad q_2 \;:=\; \frac{2\,\eps_N(r_0)}{\rho_w^{(2)}(r_0)},
\label{eq:Phi2def}
\end{equation}
with $q_2 < 1$ by (\ref{eq:cond_residL2}). As in Step~4 of Section~\ref{sec:mainproof}, $r \mapsto \Phi_2(r)/r$ is strictly increasing with $\Phi_2(r_0)/r_0 = q_2$, so $\Phi_2(r) \le q_2\, r$ on $(0, r_0]$ and $0$ is the only fixed point of $\Phi_2$ in $[0, r_0]$. With $Y_k := \sup_{t \ge k}\rho_2(t)$: for every $t \ge 0$, $\|w(\cdot,t)\|_{\Ltwo} \le \|w_0\|_{\Ltwo} + \eps_N(Y_0)$ by (\ref{eq:transportL2}) and $\|\widetilde K_N[u(\cdot,t)]\|_{\Ltwo} \le \eps_N(\rho_2(t)) \le \eps_N(Y_0)$, so $\rho_2(t) \le \LipT^{(2)}(r_0)\,\|w_0\|_{\Ltwo} + \Phi_2(Y_0)$; taking the supremum and using $\Phi_2(Y_0) \le q_2\, Y_0$,
\begin{equation}
Y_0 \;\le\; \alpha_2(s) \;:=\; \frac{1 + \Lipk^{(2)}(r_0)}{\bigl(1 - \Lipk^{(2)}(r_0)\bigr)\,(1 - q_2)}\; s,
\label{eq:alpha2def}
\end{equation}
and (\ref{eq:cond_initL2}) is precisely the condition $\alpha_2(s) < r_0$. For $t \ge k+1$, the sliding window $[t-1, t] \subseteq [k, \infty)$ gives $\|w(\cdot,t)\|_{\Ltwo} \le \eps_N(Y_k)$ and $\|\widetilde K_N[u(\cdot,t)]\|_{\Ltwo} \le \eps_N(Y_k)$, hence $Y_{k+1} \le \Phi_2(Y_k)$. The iterates $\sigma_0 := \alpha_2$, $\sigma_{k+1} := \Phi_2(\sigma_k)$ are class $\mathcal{K}$, strictly decreasing in $k$, and converge to $0$; the linear interpolation (\ref{eq:beta_def}) applied to $(\sigma_k)$, with $\sigma_{-1} := 2\sigma_0$, produces the class-$\mathcal{KL}$ envelope $\beta_2$ exactly as in Section~\ref{sec:mainproof}, proving (iv).

\emph{Step 5: geometric decay in $N$.} Identical to Step~5 of Section~\ref{sec:mainproof}, since $\eps_N$ is unchanged.
\end{proof}

\begin{remark}\label{rem:L2frames}
Two things make the $\Ltwo$ theorem work, and neither is the relocation of the residual on its own: the tail bounds of Lemma~\ref{lem:tailL2} take $\Ltwo$ arguments at no cost, their integrands being symmetric, and the Lipschitz estimates survive partial symmetrization, paying $(C_K, \rho_f) \mapsto (\sqrt{2}\,C_K,\, \rho_f/2)$. This is the whole price of the change of norm, and every constant in Theorem~\ref{thm:L2} is accordingly the more conservative one: $r_2^* \le r^*$, $\LipT^{(2)} \ge \LipT$, the residual doubles, and the transport bound is additive where the sup-norm one took a maximum. The sup-norm frame is thus the sharper of the two. It is not, however, the only one: once Lemma~\ref{lem:KLipL2} is in hand, the $\Ltwo$ theorem could also be run in the exact-transformation frame, keeping the residual at the boundary and halving it to a single $\eps_N$. The truncated frame trades that factor of two for finiteness --- $\Tcal_N$, its inverse, and $\Tcal_N[u_0]$ are all finite Volterra operators, the infinite series surviving only inside the two tail evaluations, and its target coordinates are the ones the implemented controller linearizes.
\end{remark}

\section{Conclusions}\label{sec:conclusions}

We have shown that truncating the exact infinite Volterra linearizer to any finite order preserves closed-loop stability --- in the sup-norm (Theorem~\ref{thm:main}) and in $\Ltwo$ (Theorem~\ref{thm:L2}) --- under a controller a computer can evaluate, with the exact-feedback behavior recovered as $N \to \infty$. The truncation leaves a boundary residual the exact loop does not have, and handling it demanded machinery the exact theory never needed: a Cauchy--Schwarz simplex bridge in place of the unavailable pointwise kernel bounds, a fixed-point well-posedness argument for the now genuinely coupled closed loop, and a maximal-time invariance argument with a $\mathcal{KL}$ iteration in place of the exact loop's closed-form decay.

Where \cite{krstic2026feedbacklinearizationhyperbolicpdes} gives exact exponential stabilization under an infinite controller, the present controller is finite and implementable, at the cost of practical rather than exact decay. Letting $N \to \infty$ does not recover an $\Linf$ theorem for the infinite-series feedback --- the higher-order smallness that admits truncation vanishes in the limit --- so that case remains open.

A related degree-by-degree idea appears in Krener's Al'brekht expansion for a reaction--diffusion equation \cite{9580683}, there for an optimal feedback.

The same truncation development should extend, in principle, to the parabolic Volterra construction of \cite{vazquez2008volterra1,vazquez2008volterra2}, with a semigroup-decay estimate replacing the transport-crossing argument and yielding asymptotic rather than finite-time attractivity.

Finally, $K_N$ is a single operator $(f_2, \ldots, f_N; u) \mapsto K_N[u](1)$ from plant data and state to the scalar input, and is thus --- like the linear backstepping gain learned in \cite{Bhan2024Neural} --- amenable to operator learning; replacing it by a neural surrogate while preserving properties (i)--(iv) is the subject of the companion Part~II~\cite{krstic2026feedbacklinearizationhyperbolicpdes-NO}.

\appendix

\section{Construction of the controller kernels $k_n$}\label{app:kndef}

The kernels $k_n$ in (\ref{eq:Kdef}) are constructed in \cite{krstic2026feedbacklinearizationhyperbolicpdes} by recursive integration along characteristics of a triangular cascade of first-order transport equations on the simplices $T_n(x)$; we reproduce that construction here, with $f_n$ the plant coefficients of (\ref{eq:Fdef}) and the convention $\xi_0 := x$.

For indices $p \ge 1$, $m \ge 2$, $1 \le j \le p$, define the permutation-summation operator $D_j^{p,m}$ acting on a function $g(x, \zeta_1, \ldots, \zeta_{p+m-1})$ by
\begin{eqnarray}
&& \bigl[D_j^{p,m} g\bigr](x, \zeta_1, \ldots, \zeta_{p+m-1}) \nonumber \\
&& \quad = \sum_{(\gamma_1,\ldots,\gamma_{p+m-1-j}) \in P_{p-j}(\zeta_{j+1},\ldots,\zeta_{p+m-1})} \nonumber \\
&& \quad\quad g(x, \zeta_1, \ldots, \zeta_{j-1}, \zeta_j, \gamma_1, \ldots, \gamma_{p+m-1-j}),
\label{eq:Ddef}
\end{eqnarray}
where $P_{p-j}(\zeta_{j+1}, \ldots, \zeta_{p+m-1})$ denotes the set of ordered $(p+m-1-j)$-tuples whose first $p-j$ entries are any $p-j$ elements of $\{\zeta_{j+1}, \ldots, \zeta_{p+m-1}\}$ taken in their original order, with the remaining $m-1$ entries being the leftover elements of the same set, also in their original order.

For $2 \le m \le n$, the coupling operator $B_n^m$ is then
\begin{eqnarray}
&& B_n^m[k_{n-m+1}, f_m](x, \xi_1, \ldots, \xi_n) \nonumber \\
&& \quad = \sum_{j=1}^{n-m+1} \int_{\xi_j}^{\xi_{j-1}} \bigl[D_j^{n-m+1,m} k_{n-m+1}\bigr]\bigl(x, \xi_1, \ldots, \nonumber \\
&& \quad\quad\quad \xi_{j-1}, s, \xi_j, \ldots, \xi_{n-m}\bigr) \nonumber \\
&& \quad\quad\quad \times f_m(s, \xi_{n-m+1}, \ldots, \xi_n)\, ds.
\label{eq:Bdef}
\end{eqnarray}

The kernels $k_n$ are then defined recursively. The base case is
\begin{eqnarray}
k_2(x, \xi_1, \xi_2) &=& -\int_0^{\xi_2} f_2(x - \xi_2 + s,\, \xi_1 - \xi_2 + s,\, s)\, ds, \nonumber \\
&&
\label{eq:k2def}
\end{eqnarray}
and for each $n \ge 3$,
\begin{eqnarray}
&& k_n(x, \xi_1, \ldots, \xi_n) \nonumber \\
&& = -\int_0^{\xi_n} \biggl[f_n - \sum_{m=2}^{n} B_n^m[k_{n-m+1}, f_m]\biggr]\bigl(x - \xi_n + s, \nonumber \\
&& \quad\quad \xi_1 - \xi_n + s,\, \ldots,\, \xi_{n-1} - \xi_n + s,\, s\bigr)\, ds,
\label{eq:kndef}
\end{eqnarray}
on $T_n(1)$.

The right-hand side of (\ref{eq:kndef}) at order $n$ depends only on $f_2, \ldots, f_n$ and the previously constructed $k_2, \ldots, k_{n-1}$, so the recursion is well-defined, and the $\Ltwo$ slice bound (\ref{eq:knL2}) is established from it in \cite{krstic2026feedbacklinearizationhyperbolicpdes}.

\section{Local invertibility of $\Tcal$ in sup-norm}\label{app:Tinv}

The two lemmas of this appendix establish, in the sup-norm, the local invertibility of $\Tcal = I - K$ and the Lipschitz constant of $\Tcal^{-1}$. The construction is by contraction mapping on a closed ball. Both the structure and the technique parallel the $\Ltwo$-norm versions in \cite{krstic2026feedbacklinearizationhyperbolicpdes}; only the norm changes.

\begin{lemma}[Local invertibility of $\Tcal$]\label{lem:Tinv}
Let $r_0 \in (0, r^*)$ with $r^*$ as in Section~\ref{sec:truncfeedback}, so that $\Lipk(r_0) < 1$. For every $w \in \Linf(0,1)$ with
\begin{equation}
\|w\|_\infty \;\le\; \rho_w(r_0) \;=\; r_0\,\bigl(1 - \Lipk(r_0)\bigr),
\label{eq:winball}
\end{equation}
there exists a unique $u \in \Linf(0,1)$ with $\|u\|_\infty \le r_0$ such that $\Tcal[u] = w$, equivalently $u = w + K[u]$.
\end{lemma}

\begin{proof}
Define $\Phi_w[u] := w + K[u]$ on the closed ball $\Bball_{r_0} := \{u \in \Linf(0,1) : \|u\|_\infty \le r_0\}$, which is a complete metric space under $\|\cdot\|_\infty$.

\emph{Self-map.} For $u \in \Bball_{r_0}$, applying Lemma~\ref{lem:KLip} with $u_2 = 0$ and using $K[0] = 0$,
\begin{eqnarray}
\|K[u]\|_\infty &=& \|K[u] - K[0]\|_\infty \nonumber \\
&\le& \Lipk(r_0)\,\|u\|_\infty \;\le\; \Lipk(r_0)\, r_0.
\label{eq:Ku_norm}
\end{eqnarray}
Hence
\begin{eqnarray}
\|\Phi_w[u]\|_\infty &\le& \|w\|_\infty + \Lipk(r_0)\, r_0 \nonumber \\
&\le& r_0\bigl(1 - \Lipk(r_0)\bigr) + \Lipk(r_0)\, r_0 \nonumber \\
&=& r_0,
\label{eq:Phiw_norm}
\end{eqnarray}
so $\Phi_w[u] \in \Bball_{r_0}$.

\emph{Contraction.} For $u_1, u_2 \in \Bball_{r_0}$, by Lemma~\ref{lem:KLip},
\begin{eqnarray}
\|\Phi_w[u_1] - \Phi_w[u_2]\|_\infty &=& \|K[u_1] - K[u_2]\|_\infty \nonumber \\
&\le& \Lipk(r_0)\,\|u_1 - u_2\|_\infty,
\label{eq:Phiw_contract}
\end{eqnarray}
with $\Lipk(r_0) < 1$ by the choice $r_0 < r^*$.

By the Banach fixed-point theorem, $\Phi_w$ has a unique fixed point in $\Bball_{r_0}$.
\end{proof}

Continuous dependence of the inverse on its argument is built into the contraction. The next lemma extracts the Lipschitz constant.

\begin{lemma}[Sup-norm Lipschitzness of $\Tcal^{-1}$]\label{lem:TinvLip}
Under the hypotheses of Lemma~\ref{lem:Tinv}, the inverse mapping $\Tcal^{-1} : \Bball_{\rho_w(r_0)} \to \Bball_{r_0}$ defined by $\Tcal^{-1}[w] :=$ (the unique $u$ produced by Lemma~\ref{lem:Tinv}) satisfies, for every $w_1, w_2 \in \Bball_{\rho_w(r_0)}$,
\begin{equation}
\|\Tcal^{-1}[w_1] - \Tcal^{-1}[w_2]\|_\infty \;\le\; \LipT(r_0)\,\|w_1 - w_2\|_\infty,
\label{eq:TinvLip}
\end{equation}
with $\LipT(r_0)$ given by (\ref{eq:LipTinvdef}).
\end{lemma}

\begin{proof}
Let $u_i = \Tcal^{-1}[w_i]$, so that $u_i \in \Bball_{r_0}$ and $u_i = w_i + K[u_i]$ for $i = 1, 2$. Subtracting,
\begin{equation}
u_1 - u_2 \;=\; (w_1 - w_2) + \bigl(K[u_1] - K[u_2]\bigr).
\label{eq:u1u2_diff}
\end{equation}
By the triangle inequality and Lemma~\ref{lem:KLip},
\begin{equation}
\|u_1 - u_2\|_\infty \;\le\; \|w_1 - w_2\|_\infty + \Lipk(r_0)\,\|u_1 - u_2\|_\infty.
\label{eq:u1u2_bound}
\end{equation}
Since $\Lipk(r_0) < 1$, rearranging gives (\ref{eq:TinvLip}).
\end{proof}

\section{Plant and truncated-controller bounds in sup-norm}\label{app:Fbounds}

The two lemmas of this appendix justify the four constants $M_F$, $L_F$, $B_{K_N}$, $\LipKn$ defined in Section~\ref{sec:wellposed} as the sup-norm bounds and Lipschitz constants of $F$ and $K_N$, respectively. Both derivations are direct geometric summations from Assumption~\ref{ass:fnbound} (for $F$) and from the kernel bound of Proposition~\ref{prop:L2bound} (for $K_N$).

\begin{lemma}[Plant nonlinearity bounds]\label{lem:Fbounds}
Under Assumption~\ref{ass:fnbound}, for every $u, v \in \Linf(0,1)$ with $\|u\|_\infty, \|v\|_\infty \le R < \rho_f$,
\begin{align}
\|F[u]\|_\infty &\le M_F(R), \label{eq:Fbound}\\
\|F[u] - F[v]\|_\infty &\le L_F(R)\,\|u - v\|_\infty. \label{eq:Flip}
\end{align}
\end{lemma}

\begin{proof}
The simplex $T_n(x) \subseteq T_n(1)$ has Lebesgue measure $|T_n(x)| = x^n/n! \le 1/n!$, so by Assumption~\ref{ass:fnbound}, for $u \in \overline B_R$,
\begin{eqnarray}
\|F[u]\|_\infty &\le & \sum_{n=2}^\infty \|f_n\|_\infty\,\frac{1}{n!}\,\|u\|_\infty^n \nonumber\\
&\le& \sum_{n=2}^\infty \frac{D_f}{\rho_f^{n-1}}\,R^n \nonumber\\
&=& D_f\,\frac{R^2}{\rho_f - R} \;=\; M_F(R),
\label{eq:Fbound_proof}
\end{eqnarray}
the geometric series converging since $R/\rho_f < 1$. For (\ref{eq:Flip}), the multilinear telescoping identity $u^{\otimes n} - v^{\otimes n} = \sum_{k=0}^{n-1} u^{\otimes k} \otimes (u-v) \otimes v^{\otimes(n-1-k)}$ on $T_n(x)$ gives, for $u, v \in \overline B_R$,
\begin{eqnarray}
\|F[u] - F[v]\|_\infty &\le & \sum_{n=2}^\infty \|f_n\|_\infty\,\frac{1}{n!}\,n\,R^{n-1}\,\|u - v\|_\infty \nonumber\\
&\le & D_f\,\|u - v\|_\infty\,\sum_{n=2}^\infty n\,\bigl(R/\rho_f\bigr)^{n-1} \nonumber\\
&=& L_F(R)\,\|u - v\|_\infty,
\label{eq:Flip_proof}
\end{eqnarray}
using $\sum_{n=2}^\infty n\,x^{n-1} = (1-x)^{-2} - 1$ for $|x| < 1$.
\end{proof}

The truncated controller is itself a partial Volterra series, so the same machinery gives its sup-norm bound and Lipschitz constant directly --- now with finite sums in place of the infinite ones.

\begin{lemma}[Truncated controller bounds]\label{lem:KNbounds}
Under Assumption~\ref{ass:fnbound} and Proposition~\ref{prop:L2bound}, for every $u, v \in \Linf(0,1)$ with $\|u\|_\infty, \|v\|_\infty \le R < 1/C_K$ and every $N \ge 2$,
\begin{align}
|K_N[u](1)| &\le B_{K_N}(R), \label{eq:KNbound}\\
|K_N[u](1) - K_N[v](1)| &\le \LipKn(R)\,\|u - v\|_\infty. \label{eq:KNlip}
\end{align}
\end{lemma}

\begin{proof}
By Lemma~\ref{lem:pointwise_tail} restricted to $2 \le n \le N$, $\bigl|\int_{T_n(1)} k_n\,u^{\otimes n}\,d\xi\bigr| \le (D_K\,e^{\Upsilon_K}/C_K)(C_K R)^n$ on $\overline B_R$. Summing,
\begin{eqnarray}
|K_N[u](1)| &\le & \sum_{n=2}^N \frac{D_K\,e^{\Upsilon_K}}{C_K}\,(C_K R)^n \nonumber\\
&=& D_K\,e^{\Upsilon_K}\,R\,\sum_{n=2}^N (C_K R)^{n-1} \nonumber \\
&=& B_{K_N}(R),
\label{eq:KNbound_proof}
\end{eqnarray}
proving (\ref{eq:KNbound}). For (\ref{eq:KNlip}), the multilinear telescoping identity applied term by term, combined with the kernel bound from Proposition~\ref{prop:L2bound} and Cauchy--Schwarz on $T_n(1)$ as in the proof of Lemma~\ref{lem:pointwise_tail}, gives
\begin{eqnarray}
|K_N[u](1) - K_N[v](1)| &\le & \sum_{n=2}^N D_K\,e^{\Upsilon_K}\,n\,(C_K R)^{n-1} \nonumber \\
&& \quad \times \|u - v\|_\infty \nonumber\\
&=& \LipKn(R)\,\|u - v\|_\infty,
\label{eq:KNlip_proof}
\end{eqnarray}
the factor $n$ arising from the $n$ summands in the telescoping identity.
\end{proof}

\section{The $\Ltwo$ estimates of Section~\ref{sec:L2}}\label{app:L2}

This appendix proves the partially symmetrized Lipschitz bound of Lemma~\ref{lem:KLipL2} and collects the $\Ltwo$ plant Lipschitz constant and the invertibility of $\Tcal_N$.

\paragraph*{\normalfont\em Proof of Lemma~\ref{lem:KLipL2}.}
Fix $x \in [0,1]$. By the telescoping identity (\ref{eq:telescope}), the order-$n$ contribution to $K[u_1](x) - K[u_2](x)$ is $\sum_{j=1}^n \int_{T_n(x)} k_n(x,\xi)\, P_j(\xi)\, d\xi$ with
\begin{equation}
P_j(\xi) \;:=\; \prod_{i<j} u_1(\xi_i)\;\bigl[u_1 - u_2\bigr](\xi_j)\;\prod_{i>j} u_2(\xi_i).
\label{eq:Pjdef}
\end{equation}
The squared integrand $|P_j|^2$ is invariant under permutations of the arguments carrying $u_1$, namely $(\xi_1,\ldots,\xi_{j-1})$, and, separately, of those carrying $u_2$, namely $(\xi_{j+1},\ldots,\xi_n)$. The simplex is contained in the partially ordered region
\begin{eqnarray}
D_j(x) &:=& \bigl\{\xi \in [0,x]^n : \xi_1 \ge \cdots \ge \xi_{j-1}, \nonumber \\
&& \qquad\qquad\ \ \xi_{j+1} \ge \cdots \ge \xi_n\bigr\} \;\supseteq\; T_n(x),
\label{eq:Djdef}
\end{eqnarray}
obtained by discarding the order constraints that couple the two groups and $\xi_j$. Partial symmetrization --- over each group of like-function arguments separately --- gives
\begin{eqnarray}
\int_{T_n(x)} |P_j|^2\, d\xi &\le& \int_{D_j(x)} |P_j|^2\, d\xi \nonumber \\
&=& \frac{1}{(j-1)!\,(n-j)!} \int_{[0,x]^n} |P_j|^2\, d\xi \nonumber \\
&\le& \frac{r^{2(n-1)}\;\|u_1 - u_2\|_{\Ltwo}^2}{(j-1)!\,(n-j)!}.
\label{eq:block_bound}
\end{eqnarray}
Cauchy--Schwarz against (\ref{eq:knL2}) then bounds the order-$n$ contribution by
\begin{eqnarray}
&& \sqrt{n!}\, D_K\, C_K^{n-1}\, e^{\Upsilon_K}\, r^{n-1}\, \|u_1 - u_2\|_{\Ltwo} \nonumber \\
&& \qquad\qquad \times \sum_{j=1}^n \frac{1}{\sqrt{(j-1)!\,(n-j)!}}.
\label{eq:ordern_raw}
\end{eqnarray}
Since $n!/\bigl((j-1)!\,(n-j)!\bigr) = n \binom{n-1}{j-1}$, writing $S_n := \sum_{j=1}^n \sqrt{n!/((j-1)!\,(n-j)!)}$ and applying Cauchy--Schwarz on the index $j$ gives
\begin{eqnarray}
S_n &=& \sqrt{n}\, \sum_{i=0}^{n-1} \sqrt{\binom{n-1}{i}} \nonumber \\
&\le& \sqrt{n}\,\sqrt{n\, 2^{n-1}} \;=\; n\,\bigl(\sqrt{2}\,\bigr)^{n-1},
\label{eq:comb_bound}
\end{eqnarray}
so the order-$n$ contribution is at most
\begin{equation}
D_K\, e^{\Upsilon_K}\, n\, \bigl(\sqrt{2}\,C_K r\bigr)^{n-1}\, \|u_1 - u_2\|_{\Ltwo},
\label{eq:ordern_final}
\end{equation}
uniformly in~$x$. With $y := \sqrt{2}\,C_K r < 1$, summing $\sum_{n \ge 2} n\, y^{n-1} = y(2-y)/(1-y)^2$ gives the bound in $\|\cdot\|_\infty$, hence in $\|\cdot\|_{\Ltwo}$ on the unit interval, which is~(\ref{eq:KLipL2}); stopping at $n = N$ gives the $K_N$ statement.

\paragraph*{\normalfont\em Plant Lipschitzness and invertibility of $\Tcal_N$.}
The plant pays a factor $2$ where the controller pays only $\sqrt{2}$: the plant kernels are bounded pointwise and are paired with the telescoped products in $L^1$, so the interleaving count $2^{n-1}$ enters at full strength, while the controller kernels are available only in $\Ltwo$ and are paired in $\Ltwo$, where the count enters under a square root.

\begin{lemma}[$\Ltwo$ plant Lipschitzness]\label{lem:FLipL2}
Under Assumption~\ref{ass:fnbound}, for every $u, v \in \Ltwo(0,1)$ with $\|u\|_{\Ltwo}, \|v\|_{\Ltwo} \le R < \rho_f/2$,
\begin{equation}
\|F[u] - F[v]\|_\infty \;\le\; L_F^{(2)}(R)\,\|u - v\|_{\Ltwo},
\label{eq:FLipL2}
\end{equation}
where
\begin{equation}
L_F^{(2)}(R) \;:=\; D_f\!\left(\frac{1}{(1 - 2R/\rho_f)^2} - 1\right),
\label{eq:LF2def}
\end{equation}
i.e., (\ref{eq:LFdef}) with $\rho_f$ replaced by $\rho_f/2$.
\end{lemma}

\begin{proof}
Apply (\ref{eq:telescope}) and integrate $|P_j|$ (with $u_1 = u$, $u_2 = v$) over $T_n(x) \subseteq D_j(x)$: the same partial symmetrization, now in $L^1$, gives
\begin{eqnarray}
\int_{T_n(x)} |P_j|\, d\xi &\le& \frac{\|u\|_{L^1}^{j-1}\; \|u - v\|_{L^1}\; \|v\|_{L^1}^{n-j}}{(j-1)!\,(n-j)!} \nonumber \\
&\le& \frac{R^{n-1}\,\|u - v\|_{\Ltwo}}{(j-1)!\,(n-j)!}.
\label{eq:blockL1}
\end{eqnarray}
The bound $\|f_n\|_\infty \le n!\, D_f/\rho_f^{n-1}$ of Assumption~\ref{ass:fnbound} and the count $\sum_{j=1}^n n!/\bigl((j-1)!\,(n-j)!\bigr) = n\, 2^{n-1}$ make the order-$n$ contribution at most $D_f\, n\, (2R/\rho_f)^{n-1}\, \|u - v\|_{\Ltwo}$, uniformly in~$x$. Summing over $n \ge 2$ gives~(\ref{eq:FLipL2}).
\end{proof}

\begin{lemma}[$\Ltwo$ invertibility of $\Tcal_N$]\label{lem:TNinv}
Let $r_0 \in (0, r_2^*)$ and $N \ge 2$. For every $w \in \Ltwo(0,1)$ with $\|w\|_{\Ltwo} \le \rho_w^{(2)}(r_0)$, there exists a unique $u \in \Ltwo(0,1)$ with $\|u\|_{\Ltwo} \le r_0$ such that $\Tcal_N[u] = w$. The map $\Tcal_N^{-1}$ so defined satisfies $\Tcal_N^{-1}[0] = 0$ and, for all $w_1, w_2$ in the ball,
\begin{equation}
\bigl\|\Tcal_N^{-1}[w_1] - \Tcal_N^{-1}[w_2]\bigr\|_{\Ltwo} \;\le\; \LipT^{(2)}(r_0)\, \|w_1 - w_2\|_{\Ltwo},
\label{eq:TNinvLip}
\end{equation}
uniformly in $N$.
\end{lemma}

\begin{proof}
The contraction-mapping construction of Appendix~\ref{app:Tinv} applies verbatim on the closed $\Ltwo$ ball of radius $r_0$, with Lemma~\ref{lem:KLipL2} in place of Lemma~\ref{lem:KLip} and the contraction factor $\LipKn^{(2)}(r_0) \le \Lipk^{(2)}(r_0) < 1$; the uniformity in $N$ of the radius $\rho_w^{(2)}(r_0)$ and of the constant $\LipT^{(2)}(r_0)$ follows from the same domination of the partial sums by the full series.
\end{proof}

\section*{Dedication}

This paper is dedicated to Professor Arthur Krener. His Volterra-series-in-time work \cite{1101898} in input-output representation inspired the author to consider Volterra series representation of spatially causal nonlinearities for PDEs in the early 2000s. More recently, his work on truncated nonlinear representations in the nonlinear ODE setting \cite{Krener1984ApproximateLinearization} and in the PDE setting \cite{9580683} inspired the author to take on the task of analytical guarantees of truncated implementations of nonlinear PDE controllers.

\bibliographystyle{plain}
\bibliography{bib-hyp-FL}

\end{document}